\documentclass[letterpaper, 10 pt, conference]{ieeeconf}
\IEEEoverridecommandlockouts                              

\usepackage[utf8]{inputenc}
\usepackage[english]{babel}
\usepackage{cite}
\usepackage{graphicx}
 \graphicspath{{./figures/}}
\usepackage[justification=centering]{caption}
\usepackage{subcaption}
\usepackage{mathtools}
\usepackage{amsmath}

\allowdisplaybreaks
\usepackage{amsthm,amsfonts,amstext,amssymb}
\usepackage{dsfont}
\let\labelindent\relax
\usepackage{enumitem}   
\usepackage{indentfirst}
\usepackage[font=small,skip=0pt]{caption}
\usepackage{float}
\usepackage{xcolor}
\usepackage{xr}
\usepackage{tikz}
\usetikzlibrary{shapes,arrows,positioning,calc}
\usepackage{standalone}
\usepackage[pdftex]{hyperref}
\usepackage[normalem]{ulem}
\usepackage{algorithm}
\usepackage{algpseudocode}

\newtheorem{theorem}{Theorem}
\newtheorem{corollary}{Corollary}

\theoremstyle{definition}
\newtheorem{definition}{Definition}
\theoremstyle{remark}
\newtheorem{remark}{Remark}
\newtheorem{assumption}{Assumption}
\newtheorem{example}{Example}

\algrenewcommand\algorithmicrequire{\textbf{Input:}}
\newcommand{\mpcommentout}[1]{}

\newcommand{\ra}{\rightarrow}

\newcommand{\Z}{\mathbb{Z}}
\newcommand{\N}{\mathbb{N}}
\newcommand{\E}[1]{\mathbb{E}\left[ #1 \right]}

\newcommand{\takeofftau}{\tau}

\newcommand{\cruisetau}{\tau_c}

\newcommand{\vertset}{V}
\newcommand{\vertcard}[1]{|V_{#1}|}
\newcommand{\routeset}{P}
\newcommand{\routecard}{|P|}
\newcommand{\aircraftset}{A}
\newcommand{\aircraftcard}{|A|}
\newcommand{\slotset}[1]{S_{#1}}

\newcommand{\routingvec}[2]{r_{#2}^{#1}}
\newcommand{\numsafesch}{|R|}
\newcommand{\safeschset}{R}
\newcommand{\Cyclelength}{T_{\text{cyc}}}

\newcommand{\Arrivalrate}[1]{\lambda_{#1}}
\newcommand{\Queuelength}[1]{Q_{#1}}

\newcommand{\Lyap}[1]{f(#1)}

\begin{document}
\title{\LARGE \bf
Throughput Maximizing Takeoff Scheduling for eVTOL Vehicles in On-Demand Urban Air Mobility Systems}
\author{Milad~Pooladsanj$^{1}$,
        Ketan~Savla$^{2}$,
        and Petros A. Ioannou$^{1}$
\thanks{$^{1}$M. Pooladsanj and P. Ioannou are with the Department of Electrical Engineering, University of Southern California, Los Angeles, CA 90007 USA {\tt\small (email: pooladsa@usc.edu; ioannou@usc.edu)}.}
\thanks{$^{2}$K. Savla is with the Departments of Civil and Environmental Engineering, Electrical and Computer Engineering, and the Industrial and Systems Engineering, University of Southern California, Los Angeles CA 90089 USA {\tt\small (email: ksavla@usc.edu)}. K. Savla has financial interest in Xtelligent, Inc.}
\thanks{
This research was supported in part by the PWICE institute at USC.
}
}\maketitle

\begin{abstract}
Urban Air Mobility (UAM) offers a solution to current traffic congestion by using electric Vertical Takeoff and Landing (eVTOL) vehicles to provide on-demand air mobility in urban areas. Effective traffic management is crucial for efficient operation of UAM systems, especially for high-demand scenarios. In this paper, we present a centralized framework for conflict-free takeoff scheduling of eVTOLs in on-demand UAM systems. Specifically, we provide a scheduling policy, called VertiSync, which jointly schedules UAM vehicles for servicing trip requests and rebalancing, subject to safety margins and energy requirements. We characterize the system-level throughput of VertiSync, which determines the demand threshold at which the average waiting time transitions from being stable to being increasing over time. We show that the proposed policy maximizes throughput for sufficiently large fleet size and if the UAM network has a certain symmetry property. We demonstrate the performance of VertiSync through a case study for the city of Los Angeles, and show that it significantly reduces average passenger waiting time compared to a first-come first-serve scheduling policy. 
\end{abstract}

\section{Introduction}
Traffic congestion is a significant problem in urban areas, leading to increased travel times, reduced productivity, and environmental concerns. A potential solution to this problem is to add another mode of transportation, such as Urban Air Mobility (UAM), which aims to use urban airspace for on-demand mobility \cite{holden2016uber}. Although the idea of using flying vehicles in urban areas for transportation purposes is not new, e.g., the use of helicopters from 1940s to 1970s~\cite{safadi2023macroscopic}, it has re-gained attention both in academia and industry due to large investments in electric Vertical Takeoff and Landing (eVTOL) vehicles~\cite{karp2022gambling}. While recent research efforts has focused mostly on the vehicles themselves, there has been limited attention paid to the question of how a potentially large fleet of UAM vehicles should operate collectively, and how the traffic should be managed \cite{mueller2017enabling, chin2023traffic}. The objective of traffic management is to efficiently use the limited UAM resources, such as the airspace, takeoff and landing areas, and the UAM vehicles, to meet the demand. The purpose of this paper is to systematically design and analyze an air traffic management protocol for on-demand UAM systems.
\par
The starting point of this paper is the well-known Air Traffic Flow Management (ATFM) problem for commercial airplanes, which was formalized in~\cite{richetta1993solving}. The key idea behind ATFM is to proactively manage congestion by anticipating traffic demand and manage the usage of various airspace and airport resources. To this end, an integer program formulation, called the Traffic Flow Management Problem (TFMP), is solved to assign desired flight trajectories to each airplane subject to operational constraints \cite{bertsimas1998air, bertsimas2011integer}. The UAM traffic management problem can be considered a natural extension of ATFM and its integer program formulation. However, unlike commercial air traffic, where the demand is predictable even weeks in advance, the UAM systems will be designed to provide on-demand services. This poses a significant tactical challenge. 
\par
\begin{figure}[t]
    \centering
    \includegraphics[width=0.4\textwidth]{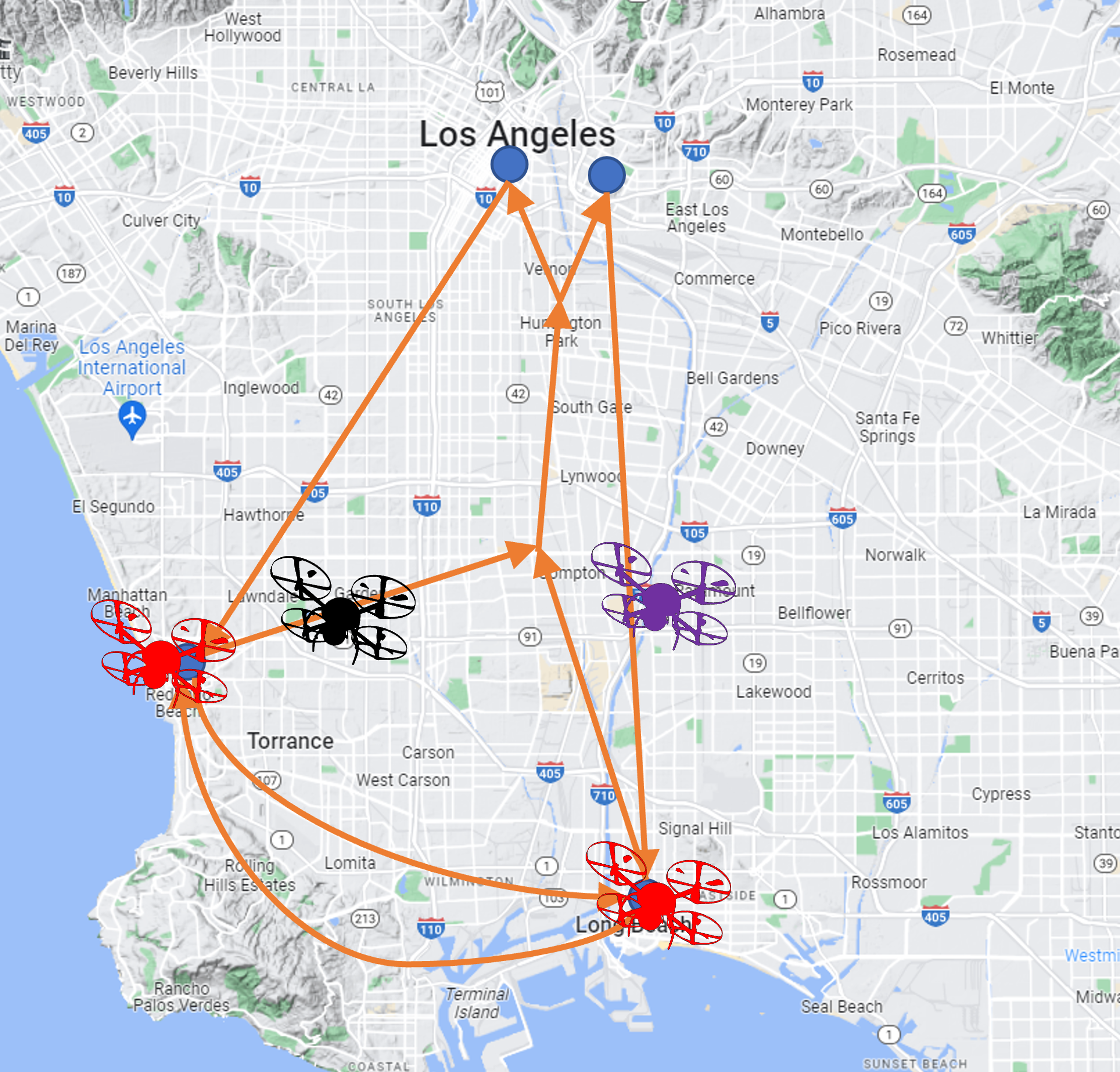}
    
    \vspace{0.1 cm}
    
    \caption{\sf A top-view sketch of a UAM network with three modes of UAM vehicle operation: idle vehicle (red), in-service vehicle that transports passengers (black), and rebalancing vehicle that flies without passengers to high-demand areas (purple).}
    \label{fig:rebalancing}
\end{figure}
The UAM traffic management methods can be generally classified as either centralized or decentralized \cite{bauranov2021designing}. In decentralized methods, each UAM vehicle can select its preferred route while being responsible for its separation margins with other vehicles using onboard technology. These methods can be based on cooperative multi-agent negotiation \cite{wollkind2004automated}, mixed-integer linear
programming \cite{schouwenaars2004decentralized}, decentralized model predictive control \cite{richards2004decentralized}, or Markov decision processes \cite{yang2021autonomous}.
\par
While it is reasonable to expect that the decentralized approach is feasible for low-volume UAM operation, it can lead to a significant loss in efficiency, and even gridlock \cite{FAAConOps}. For this and safety reasons, and for the likely scenario in the near future, it is natural to focus on centralized UAM traffic management, where a central authority assigns conflict-free flight trajectories to each UAM vehicle. Another reason for using centralized traffic management is the limited battery capacity of UAM vehicles \cite{thipphavong2018urban}. Since UAM vehicles operate on electric power, their flight range and operational flexibility are constrained by battery limitations. Efficient scheduling and routing minimizes unnecessary energy consumption, ensures timely recharging, and prevents mid-air energy depletion. Centralized methods are typically modeled as an optimization problem. In \cite{zhu2019pre}, the authors consider a two-phase approach, where in the first phase, an integer program is solved to determine a conflict-free trajectory for a given flight. The solution of the optimization problem is then passed to a velocity profile smoothing model in the second phase. The work in \cite{tang2021automated} considers a combination of pre-computed optimal paths and integer program to determine conflict-free trajectories for all flight requests. Recent works such as \cite{chin2023protocol} extend the existing TFMP formulation to accommodate the on-demand nature of UAM and incorporate fairness. Solution methods other than optimization formulations include heuristic methods such as first-come first-served scheduling \cite{pradeep2018heuristic} and iterative conflict detection and resolution \cite{bosson2018simulation}.
\par
The previous works provide valuable insights into the operation of UAM systems. However, most, if not all, of them do not explicitly address two critical aspects. First is the concept of \emph{rebalancing}: the UAM vehicles will need to be constantly redistributed in the network when the demand for some destinations is higher than others; see Figure~\ref{fig:rebalancing}. Efficient rebalancing ensures the effectiveness and sustainability of on-demand UAM systems. The concept of rebalancing has been explored extensively in the context of on-demand ground transportation \cite{pavone2012robotic}. However, these studies predominantly use flow-level formulations which do not capture the safety and separation considerations associated with UAM vehicle operations. The second aspect which has not been addressed in the UAM literature is a thorough characterization of the system-level throughput. Roughly, the throughput of a given traffic management policy determines the highest demand that the policy can handle \cite{POOLADSANJ2023TRC}. In the context of UAM, the throughput has strong implications on average passenger waiting time. In particular, throughput determines the demand threshold at which the average passenger waiting time transitions from being stable to being increasing over time. Therefore, it is desirable to design a policy that achieves the maximum possible throughput.
\par
In response to the aforementioned gaps in the literature, we propose a centralized policy, called VertiSync, for conflict-free takeoff scheduling in on-demand UAM networks. VertiSync jointly schedules the UAM vehicles for servicing trip requests and rebalancing within the network, subject to safety margins and energy requirements. Since our primary focus is on scheduling, the energy requirements are treated as sufficient conditions that do not constrain the scheduling process.
\par
The proposed policy modifies and extends the TFMP by incorporating elements of our recent work on ramp metering in ground transportation \cite{POOLADSANJ2023TRC}. In \cite{POOLADSANJ2023TRC}, we used queuing theory to design algorithms that maximize freeway throughput without knowledge of demand. Although the methods for characterizing throughput in this paper are conceptually similar to those in \cite{POOLADSANJ2023TRC}, the underlying systems are significantly different. Specifically, in \cite{POOLADSANJ2023TRC}, we model an entire freeway as a single static ``server". In contrast, in this paper, we deal with multiple dynamic servers—UAM vehicles moving around in the system—adding complexity to the problem. 
\par
The primary contributions of this paper are as follows:
\begin{enumerate}
    \item Developing a conflict-free takeoff scheduling policy, called VertiSync, for on-demand UAM networks, subject to safety margins and energy requirements.
    \item Explicitly incorporating rebalancing into the UAM scheduling framework.
    \item Characterizing the system-level throughput of VertiSync, and demonstrating its effectiveness through a case study for the city of Los Angeles.
\end{enumerate}
\par
The rest of the paper is organized as follows: in Section~\ref{sec:problem-formulation}, we describe the problem formulation. We provide our traffic management policy and characterize its throughput in Section~\ref{section:scheduling}. We provide the Los Angeles case study in Section~\ref{section:simulation}, and conclude the paper in Section~\ref{section:conclusion}.

\section{Problem Formulation}\label{sec:problem-formulation}
We use the following standard notations throughout the paper. We let $\N$ be the set of positive integers, and $\N_{0}$ be the set of non-negative integers. For $n \in \N$, we use $[n]$ to denote the set $\{1, 2, \cdots, n\}$. For a set $A$, we let $|A|$ denote its cardinality. 
\subsection{UAM Network Structure}
\mpcommentout{
A UAM network will consist of a number of take-off/landing areas, called \emph{vertiports}, connected by a set of routes. Each vertiport may have multiple takeoff/landing pads, called \emph{vertipads}.
}
We describe the UAM network by a directed graph $\mathcal{G}$. A node in the graph $\mathcal{G}$ represents either a \emph{vertiport}, that is, the take-off/landing area, or an intermediate point where two or more routes cross paths. A link in the graph $\mathcal{G}$ represents a section of a route (or multiple routes). We let $\vertset$ be the set of vertiports. We let $N_v$ be the total number of \emph{vertipads}, i.e., takeoff/landing pads, at vertiport $v \in \vertset$. An Origin-Destination (O-D) pair $p$ is an ordered pair $p=(o_p,d_p)$ where $o_p, d_p \in \vertset$ and there is at least one route from $o_p$ to $d_p$. We let $\routeset$ be the set of O-D pairs; see Figure~\ref{fig:graph-example}. To simplify the network representation, we assume that each vertiport has exactly one outgoing link exclusively used for takeoffs from that vertiport and a separate incoming link exclusively used for landings. To simplify the analysis, we assume that there is at most one route between any two vertiports. Although this assumption does not change our takeoff scheduling policy, future work can extend our analysis when this assumption does not hold. We also assume that the UAM routes do not conflict with any existing airspace.
\begin{figure}[t]
    \centering
    \includegraphics[width=0.4\textwidth]{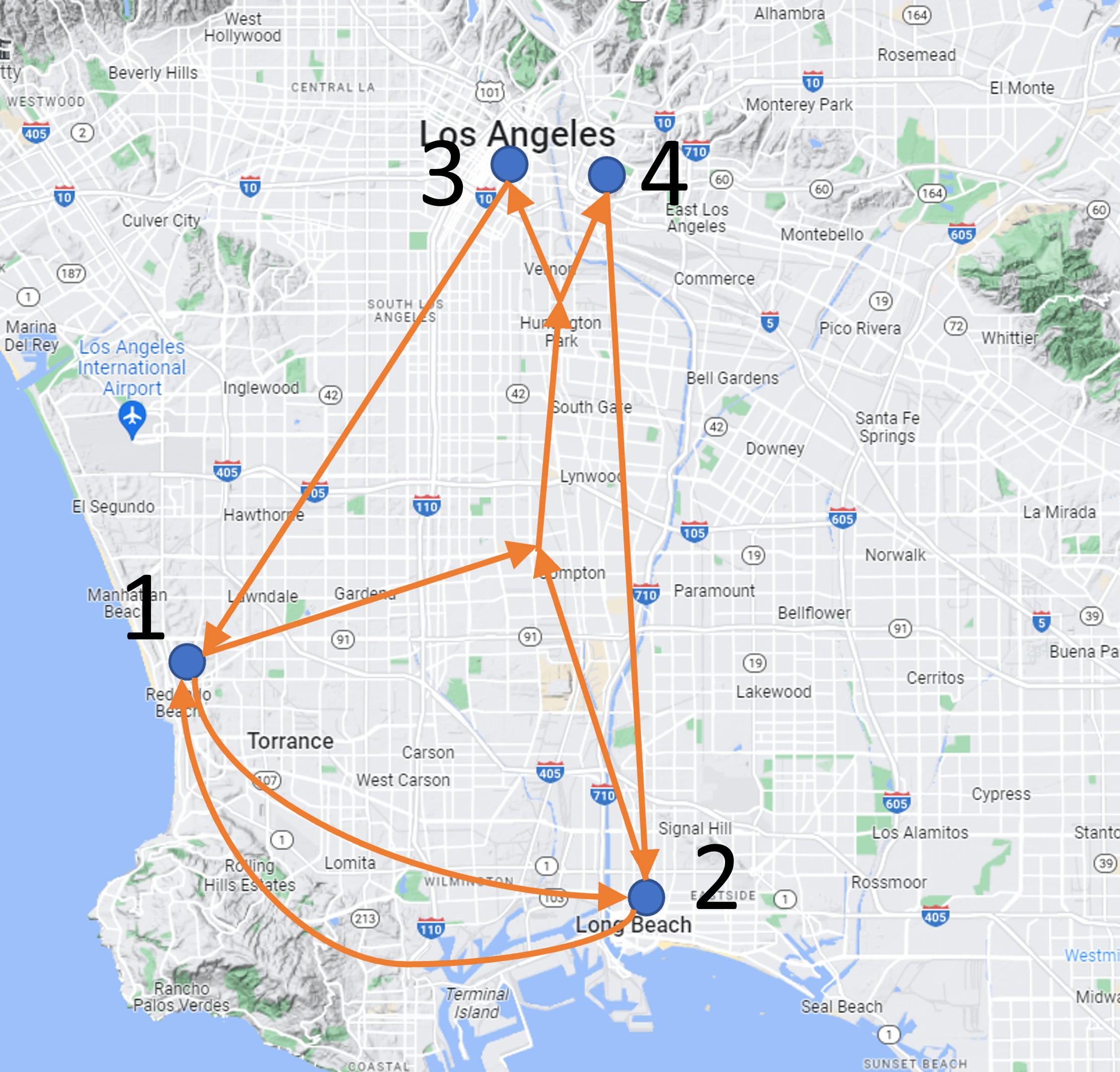}
    
    \vspace{0.1 cm}
    
    \caption{\sf A top-view sketch of a UAM network with $\vertcard{}=4$ vertiports (blue circles) and $\routecard=8$ O-D pairs $\routeset = \{(1,3), (1,4), (2,3), (2,4), (3,1), (4,2), (1,2), (2,1)\}$.}
    \label{fig:graph-example}
\end{figure}
\begin{assumption}
Given an O-D pair $p = (o_p, d_p)$, the opposite pair $q = (d_p, o_p)$ may not necessarily be an O-D pair. However, to enable rebalancing, we assume that there exists a sequence of routes connecting $d_p$ to $o_p$, where the first route originates from $d_p$ and the last route ends at $o_p$. This assumption holds for any two vertiports. In graph theory terms, this implies that the graph $\mathcal{G}$ is strongly connected. Additionally, if $q$ is not an O-D pair, we assume that UAM vehicles may land and takeoff at any intermediate vertiports along the sequence of routes. Note that this assumption only applies if $q$ is not an O-D pair and the UAM vehicle needs to be rebalanced to $o_p$. Therefore, the temporary landing and takeoff at intermediate vertiports, which might be necessitated by the scheduling algorithm for the sake of efficiency,  will be without any passengers onboard. 
\end{assumption}
\par
In the next section, we will discuss the operational constraints of UAM vehicles.

\subsection{Operational Constraints}\label{section:op-constraints}
In this section, we describe the constraints and assumptions related to the flight operations of the UAM vehicle. We assume that all UAM vehicles have the same characteristics so that these constraints are the same for all of them. Let $\aircraftset$ be the set of UAM vehicles in the system. The flight operation of each vehicle consists of the following three phases: 
\begin{itemize}
    \item \textbf{takeoff:} During this phase, the UAM vehicle is positioned on a departure vertipad and passengers (if any) are boarded before the vehicle is ready for takeoff. To position the vehicle on the departure vertipad, it is either transferred from a parking space or directly lands from a different vertiport. Let $\takeofftau$ denote the \emph{takeoff separation}, which is the minimum time required between successive takeoffs from the same vertipad. In other words, the takeoff operations are completed in a $\takeofftau$-minute time window for every flight, which implies that the takeoff rate from each vertipad is at most one vehicle per $\takeofftau$ minutes. 
    \item \textbf{airborne:} To ensure safe operation, all UAM vehicles must maintain minimum horizontal and vertical safety margins from each other while airborne. According to the FAA standards \cite{FAAConOps, bauranov2021designing}, UAM corridors are divided into three-dimensional sectors, with each sector having a certain capacity. These sectors account for the time it takes for a UAM vehicle to travel from the vertiport to the cruising altitude and back from the cruising altitude to the vertiport. Therefore, they are not necessarily equidistant. Without loss of generality, we set the capacity of each sector to $1$ vehicle to avoid the need for tactical deconfliction within a sector. We discretize time into time steps of length $\cruisetau$, assuming that $\cruisetau$ is uniform across all sectors, such that in each time-step, a vehicle moves to the next sector; see Figure~\ref{fig:space-time-discretization}. This means that $\cruisetau$ also captures the time required for a vehicle to travel between the vertiport and the cruising altitude. 
    We further assume that $\takeofftau \geq \cruisetau$, i.e., the takeoff separation is more restrictive than the separation imposed by airborne safety margins, and $k_{\tau}:= \takeofftau/\cruisetau$ is integer-valued. These assumptions are based on the current technological limitations \cite{bosson2018simulation, vascik2017constraint}. Future work may extend our results as technology or regulations evolve.
    \item \textbf{landing:} Once a UAM vehicle lands, passengers (if any) are disembarked, new passengers (if any) are embarked, and the vehicle undergoes servicing if needed. Thereafter, the vehicle is either transferred to a parking space, e.g., for re-charging, or, if it has boarded new passengers or needs to be rebalanced, takes off to another vertiport. Similar to takeoff operations, we assume that the landing operations are completed within a $\takeofftau$-minute time window for every flight. That is, once a vehicle lands, the next takeoff or landing can occur after $\takeofftau$ minutes. 
    We assume that each vertiport has enough parking capacity to clear an arriving UAM from the vertipad after landing, and that these parking spots are equipped with charging stations.
%
    \end{itemize}
\mpcommentout{
\begin{remark}
The above assumptions regarding the takeoff, landing, and airborne separations are practical given the current technological limitations \cite{bosson2018simulation, vascik2017constraint}. However, our results can be generalized and are not limited to these specific assumptions.
\end{remark}
}
\par
\begin{figure}[t]
    \centering
    \includegraphics[width=0.4\textwidth]{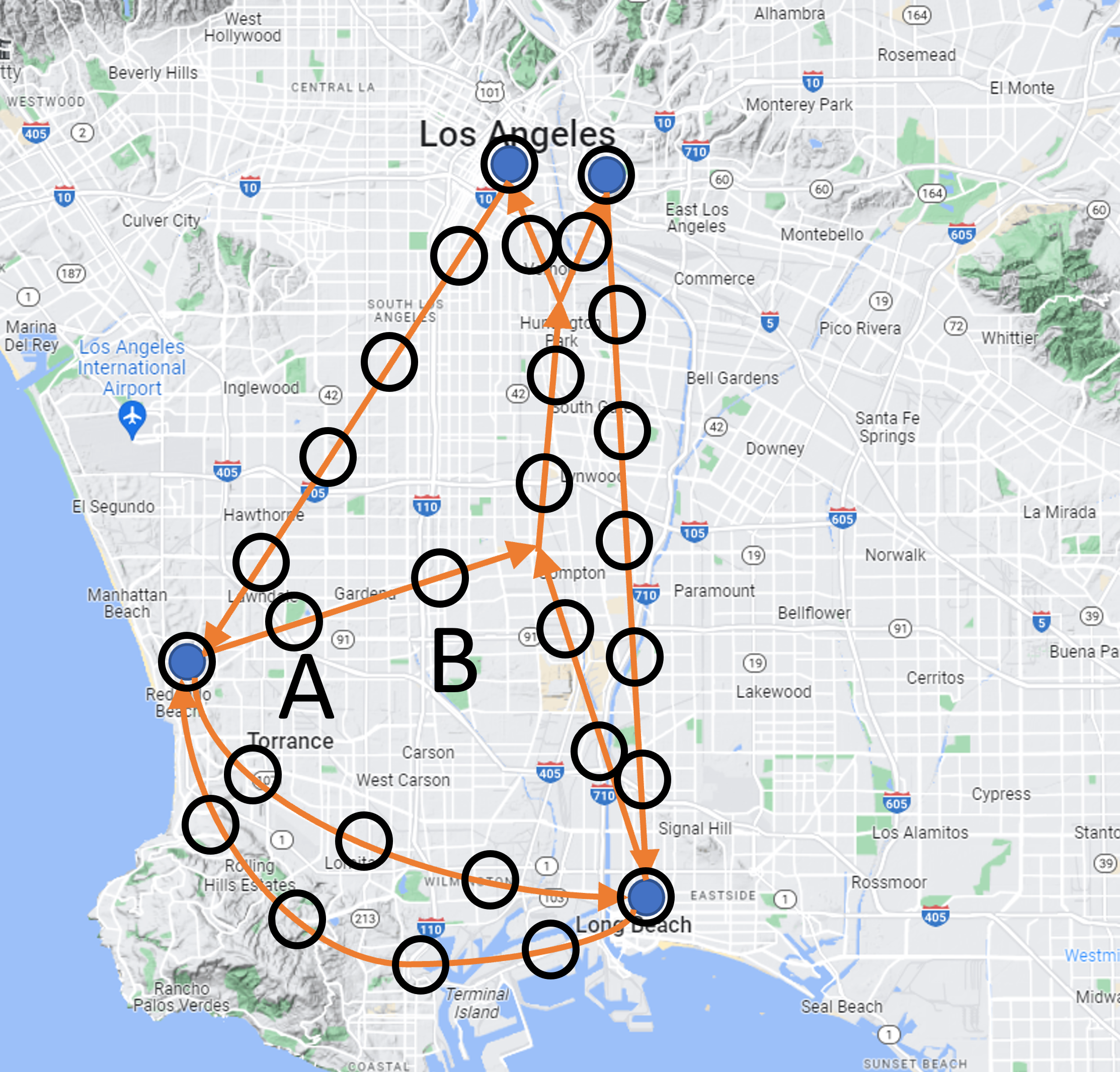}
    
    \vspace{0.1 cm}
    
    \caption{\sf Sector configuration for a UAM network, with sector capacity of $1$ vehicle, i.e., at most $1$ UAM vehicle can occupy any sector at any time. Moreover, if a UAM vehicle occupies sector A, then it moves to sector B after one time step.}
    \label{fig:space-time-discretization}
\end{figure}
Without loss of generality, we assume that different links of the graph $\mathcal{G}$ are at a safe horizontal and vertical distance from each other, except in the vicinity of the nodes where they intersect.
We consider the ideal case of no external disturbance such as adverse weather conditions. As a result, if a UAM vehicle's flight trajectory satisfies the safety margins and the separation requirements, then the vehicle follows it without deviating from the trajectory. 
Since we only focus on conflict-free takeoff scheduling policies in this paper, we do not specify the low-level vehicle controller that follows a given trajectory.

\subsection{Demand and Performance Metric}\label{section:demand}
\begin{figure}[t]
    \centering
    \includegraphics[width=0.4\textwidth]{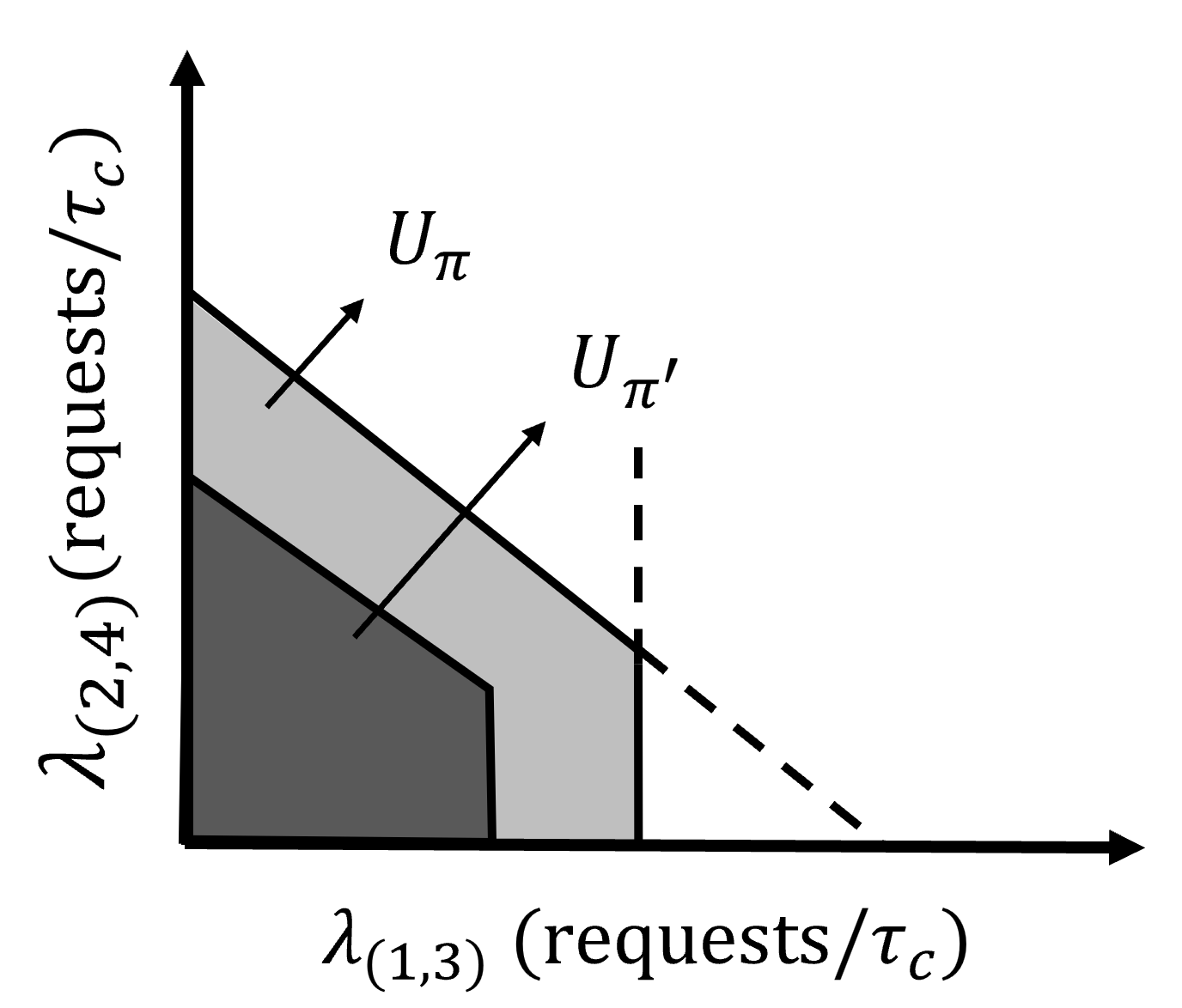}
    
    \vspace{0.1 cm}
    
    \caption{\sf An illustration of the under-saturation region of some policy $\pi'$ (dark grey area) and a throughput maximizing policy $\pi$ (dark + light grey areas).}
    \label{fig:throughput-illustration}
\end{figure}
In an on-demand UAM network, the demand is likely not known in advance. To model this unpredictability, we use exogenous stochastic processes. For ease of performance analysis, we adopt a discrete time setting. Let each time step have a duration of $\cruisetau$, which corresponds to the time needed for a UAM vehicle to move to an adjacent sector, as described in Section~\ref{section:op-constraints}. The number of trip requests for an O-D pair $p \in \routeset$ is assumed to follow an i.i.d Bernoulli process with parameter $\Arrivalrate{p}$, independent of other O-D pairs. That is, at any given time step, the probability that a new trip is requested for the O-D pair $p$ is $\Arrivalrate{p}$ independent of everything else. Note that $\Arrivalrate{p}$ specifies the rate of new requests for the O-D pair $p$ in terms of the number of requests per $\cruisetau$ minutes. Let $\Arrivalrate{} := (\Arrivalrate{p})$ denote the vector of arrival rates. 
\par
A practical scenario is one where the trip requests are made in advance by passengers, e.g., via a mobile app, with passengers arriving at the vertiport at their scheduled departure times as determined by the scheduling policy. For each O-D pair, trip requests are placed in an unlimited capacity queue until they are \emph{serviced}, at which point they exit the queue. These queues do not represent ``physical" waiting lines; rather, they function as ordered lists of trip requests based on their arrival times. 
In order for a request to be serviced, a UAM vehicle serving that request must take off from the vertiport \footnote{By using this notion of service, queue length at time $t$ is a function of takeoff at time $t$. Alternatively, one could use UAM vehicles landing as the notion of service by appropriately shifting this function in time.}. A \emph{scheduling policy} is a rule that schedules the UAM vehicles in the system for either servicing trip requests or rebalancing, i.e., taking off without passengers to service trip requests at other vertiports.
\par
The objective of the paper is to design a policy that can handle the maximum possible demand under the operational constraints discussed in Section~\ref{section:op-constraints}. The key performance metric to evaluate a policy is the notion of \emph{throughput} which we will now formalize. For $p \in \routeset$, let $\Queuelength{p}(t)$ be the number of outstanding trip requests for O-D pair $p$ at time $t$. Let $\Queuelength{}(t) = (\Queuelength{p}(t))$ be the vector of outstanding trip requests for all the O-D pairs at time $t$. We define the \emph{under-saturation} region of a policy $\pi$ as 
\begin{equation*}
    U_{\pi} = \{\Arrivalrate{}: \limsup_{t \to \infty} \E{\Queuelength{p}(t)} < \infty~~\forall p \in \routeset~ \text{under policy}~\pi\}.
\end{equation*}
\par
This is the set of $\Arrivalrate{}$'s for which the network remains \emph{under-saturated}, meaning that the expected number of outstanding trip requests remain bounded for all the O-D pairs. The boundary of this set is called the throughput of policy $\pi$. We are interested in finding a policy $\pi$ such that $U_{\pi'} \subseteq U_{\pi}$ for all policies $\pi'$, including those that have information about the demand $\Arrivalrate{}$ in advance. In other words, if the network remains under-saturated using some policy $\pi'$, then it also remains under-saturated using policy $\pi$. In that case, we say that policy $\pi$ maximizes the throughput for the UAM network. In the next section, we introduce one such policy. 
\begin{remark}
A rigorous definition of throughput should also include its dependence on the initial condition of the UAM vehicles and the initial queue sizes. We have removed this dependence for simplicity since the throughput of our policy
does not depend on the initial condition.
\end{remark}

\begin{example}\label{example:throughput-illustration}
Consider the UAM network in Figure~\ref{fig:graph-example}, and suppose that a policy $\pi$ is able to maximize the throughput; that is, for any other policy $\pi'$, we have $U_{\pi'} \subseteq U_{\pi}$. Suppose further that $\Arrivalrate{p}$ is fixed for every O-D pair $p$ except $(1,3)$ and $(2, 4)$. An illustration of $U_{\pi}$ and $U_{\pi'}$ is shown in Figure~\ref{fig:throughput-illustration}. From the figure, one can see that if $(\Arrivalrate{(1, 3)},\Arrivalrate{(2, 4)}) \in U_{\pi'}$, then $(\Arrivalrate{(1, 3)},\Arrivalrate{(2, 4)}) \in U_{\pi}$, i.e., if the expected number of outstanding trip requests remain bounded under the policy $\pi'$, then it also remains bounded under policy $\pi$. 
\end{example}

\section{Network-Wide Scheduling}\label{section:scheduling}

\mpcommentout{
\begin{figure}[t]
    \centering
    \includegraphics[width=0.3\textwidth]{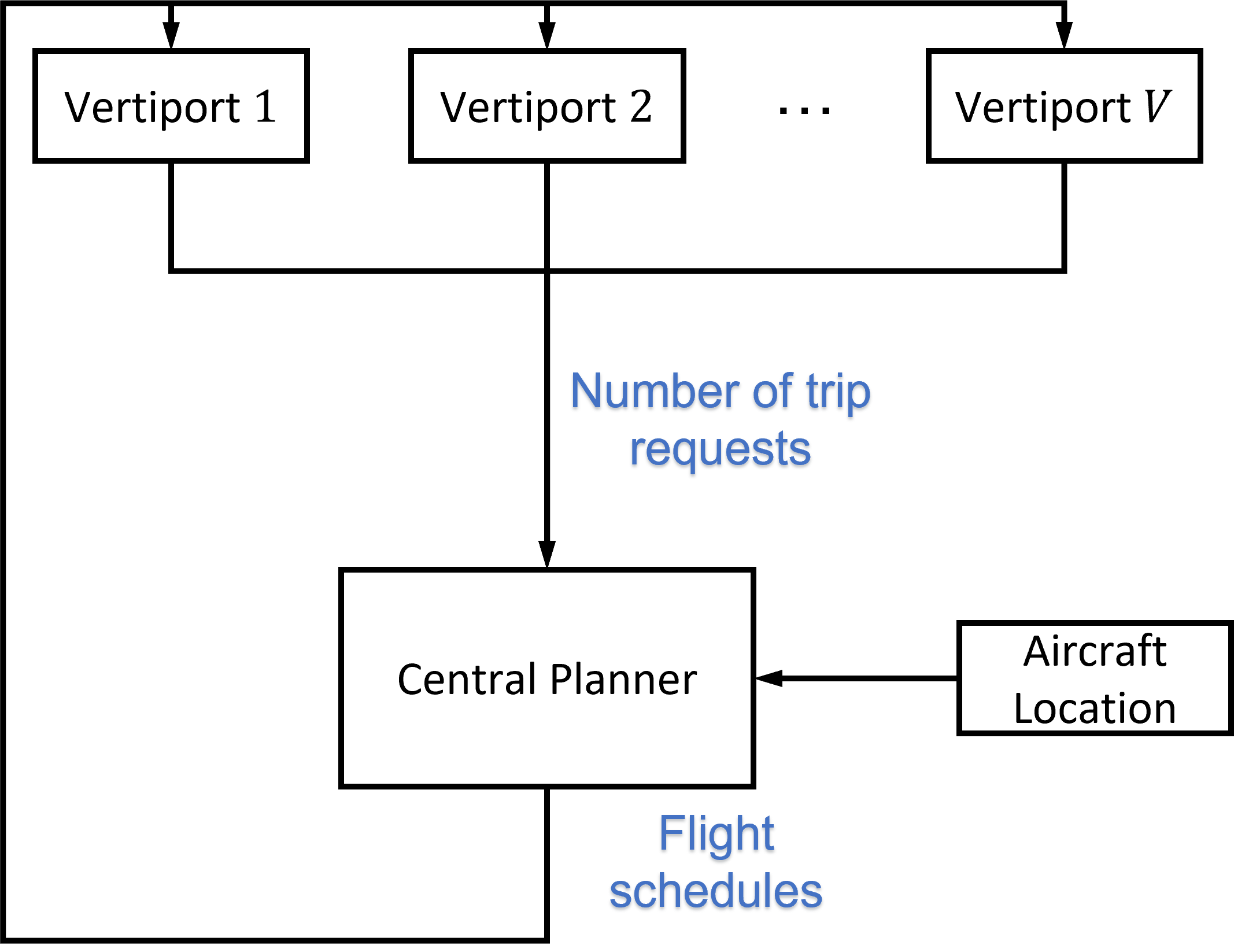}
    
    \vspace{0.1 cm}
    
    \caption{\sf The information flow in the VertiSync policy at the beginning of a cycle.}
    \label{fig:renewal-schematic}
\end{figure}
}
\subsection{VertiSync Policy}
We now introduce our scheduling policy, called VertiSync, which is inspired by the literature on queuing theory \cite{armony2003queueing} and the classical TFMP \cite{bertsimas1998air}. Informally, VertiSync works in cycles during which only the trips that were requested before the start of a cycle are serviced during that cycle. To this end, at the start of a cycle, the central planner determines conflict-free takeoff schedules to service all outstanding trip requests. Within each O-D pair, trip requests are serviced in a First Come First Serve (FCFS) manner. However, this ordering is not necessarily maintained across different O-D pairs, as the scheduling policy prioritizes conflict-free operations across the network. Once takeoff schedules are determined, they are communicated to the UAM vehicles and vertiport operators responsible for handling takeoff and landing operations at each vertiport. When all outstanding trip requests are serviced, i.e., some UAM vehicle serving those requests has taken off from the vertiport, the cycle ends, and the next cycle starts. It is assumed that the central planner knows the state of each UAM vehicle as well as the number of outstanding trip requests for each O-D pair.   
\par 
As discussed in Section~\ref{section:op-constraints}, we divide the UAM \emph{corridors} into sectors of capacity $1$, as shown in Figure~\ref{fig:slots}. Time is discretized into time steps of length $\cruisetau$, such that in each time step, a UAM vehicle moves to an adjacent sector. Recall that these sectors are constructed to ensure that adjacent sectors satisfy airborne safety margins. We extend sectors to include the origin and destination vertiports, with capacity of each \emph{vertiport} sector equal to the number of vertipads at that vertiport. 
In particular, the first sector for O-D pair $p \in \routeset$ is located at the origin vertiport $o_p$, meaning that the UAM vehicle is located on a vertipad at $o_p$, and the last sector is located at the destination vertiport $d_p$, meaning that the UAM vehicle has landed on a vertipad at $d_p$. 
Without loss of generality, we assume that if a link is common to two or more routes, then the sectors associated with those routes coincide with each other on that link, e.g., the green sector in Figure~\ref{fig:slots}. 
We assign a unique identifier to each sector, with overlapping sectors belonging to different routes having \textit{different} identifiers. For example, the green sector in Figure~\ref{fig:slots} has four different identifiers each belonging to routes $(1,3), (1,4), (2,3),$ and $(2,4)$. We do this intentionally for more clarity in the mathematical formulation of the problem later on. Let $\slotset{p}$ be the set of sectors associated with O-D pair $p$. 
\par
\begin{figure}[t]
    \centering
    \includegraphics[width=0.4\textwidth]{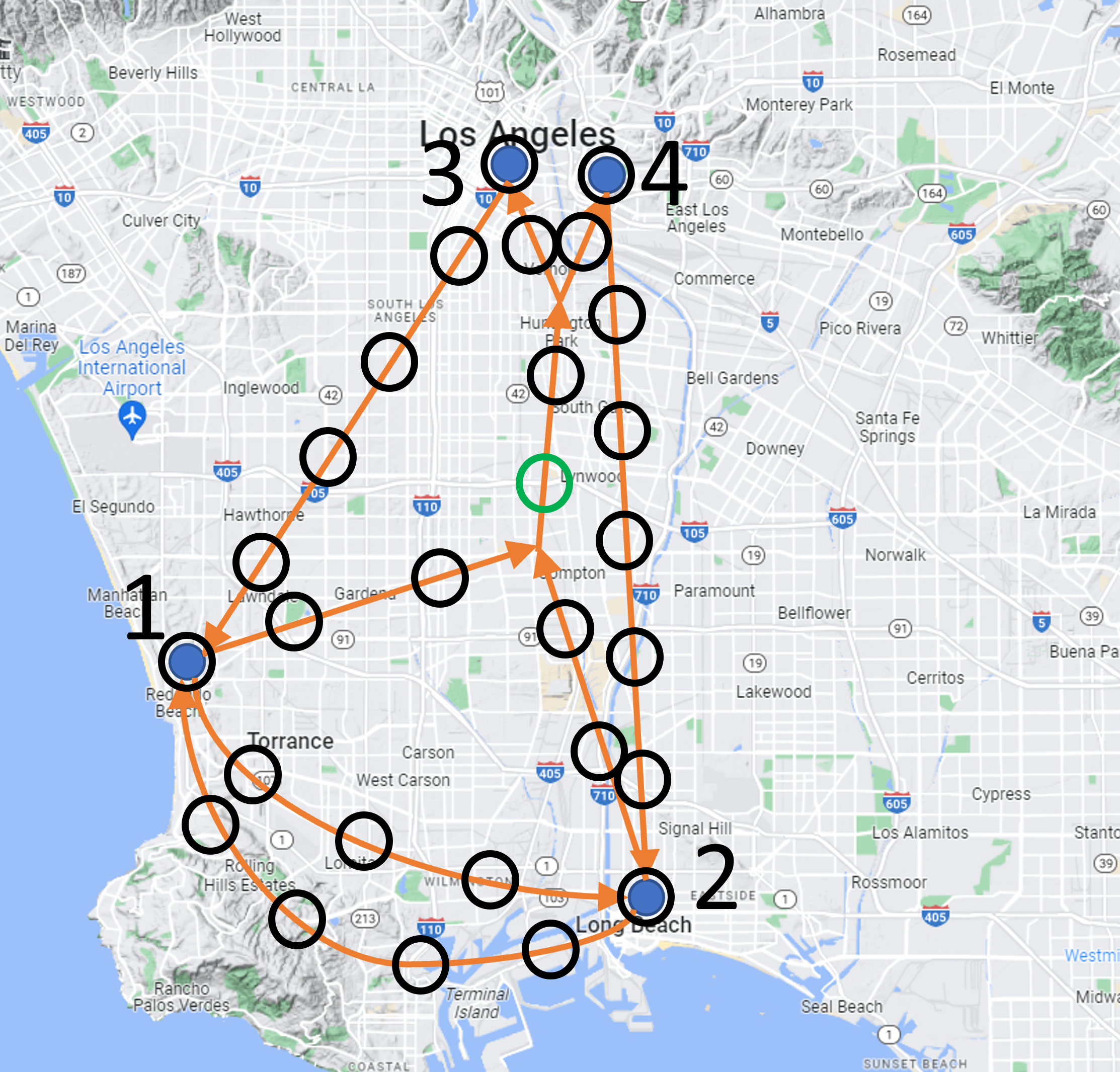}
    
    \vspace{0.1 cm}
    
    \caption{\sf Sector configuration for a UAM network. The green sector belongs to routes $(1,3), (1,4), (2,3), (2,4)$.}
    \label{fig:slots}
\end{figure} 
Let $t_k$ be the start time of the $k$-th cycle, $n \in \N_{0}$, $p \in \routeset$, and $i \in \slotset{p}$. A key decision variable in the VertiSync formulation is $w_{i,n}^{a,p}$, which represents the number of times that vehicle $a \in \aircraftset$ has visited sector $i$ of route $p$ in the time interval $(t_k,t_k+n\cruisetau]$. For brevity, $t_k$ is dropped from $w_{i,n}^{a,p}$, as we conduct scheduling one cycle at a time. By definition, $w_{i,n}^{a,p}$ is non-decreasing with respect to $n$. Moreover, if $w_{i,n}^{a,p}-w_{i,n-1}^{a,p} = 1$ for some $n \geq 1$, then it means that vehicle $a$ must have occupied sector $i$ at some time in the interval $(t_k+(n-1)\cruisetau, t_k+n\cruisetau]$. Note that this occupation time is not necessarily a multiple of $\cruisetau$. We also use the variable $w_{o,n}^{a,p}$ to represent the number of times that vehicle $a$ with route $p$ has taken off from vertiport $o_p$ in the interval $(t_k,t_k+n\cruisetau]$. Similarly, $w_{d,n}^{a,p}$ denotes the number of times that vehicle $a$ with route $p$ has landed on vertiport $d_p$ in the interval $(t_k,t_k+n\cruisetau]$. Finally, we use the variable $w_{v, n}^{a}$ to denote the number of times that vehicle $a$ has visited vertiport $v$ in the interval $(t_k,t_k+n\cruisetau]$. For sector $i$, we use the notation $i+1$ to specify the next sector along a UAM vehicle's route. Moreover, given two O-D pairs $p, q \in \routeset$ and sectors $i \in \slotset{p}$ and $j \in \slotset{q}$, we use the notation $i=j$ to specify that sector $i$ coincides with sector $j$. 
We are now in a position to formally introduce VertiSync.
\begin{definition}\label{def:renewal-policy}
\textbf{(VertiSync Policy)} 
The policy works in cycles of \emph{variable} length, with the first cycle starting at time $t_1 = 0$. At the beginning of the $k$-th cycle at time $t_k$, each vertiport communicates to the central planner the number of trip requests for each O-D pair that originates from that vertiport, i.e., the vector of trip requests $\Queuelength{}(t_k) = (\Queuelength{p}(t_k))$ is communicated to the central planner. During the $k$-th cycle, only these requests will be serviced. The $k$-th cycle ends once all these requests have been serviced, i.e., right after the last takeoff. 
\par
The central planner solves the following optimization problem to determine the takeoff schedules during the $k$-th cycle. The objective of the optimization is to minimize the total flight time of all UAM vehicles \footnote{In fact, we are minimizing the total rebalancing component of this flight time.}. That is, we minimize:
\begin{equation}\label{eq:TFMP-objective}
\sum_{a \in \aircraftset}\sum_{p \in \routeset}(w_{o,M_k}^{a,p}-w_{o,0}^{a,p})T_p,
\end{equation} 
where $T_p$ is the flight time for route $p$, including the time it takes to travel between the vertiport and the cruising altitude, and $M_k \in \N$ is such that $M_k \cruisetau$ is a conservative upper bound on the duration of the $k$-th cycle. For example, $\sum_{p \in \routeset}\Queuelength{p}(T_p + \Tilde{T}_p)$ is one such upper bound, which is calculated under the assuming that all outstanding trip requests are serviced by a single UAM vehicle. Here, $\Tilde{T}_p$ represents the travel time from $d_p$ to $o_p$. The following constraints must be satisfied:
    \begin{subequations}\label{eq:TFMP-constraints}
        \begin{align}
            &\sum_{a \in \aircraftset}(w_{o,M_k}^{a,p}-w_{o,0}^{a,p}) \geq \Queuelength{p}(t_k),\hspace{1cm} \forall p \in \routeset, \label{eq:TFMP-constraint-serviceall} \\
             &w_{i,n-1}^{a,p}-w_{i,n}^{a,p} \leq 0, \notag \\ 
             &\hspace{1.5cm}\forall n \in [M_k],~ p \in \routeset,~i \in \slotset{p},~ a \in \aircraftset, \label{eq:TFMP-constraint-w-increasing}\\
%
%
             &w_{i+1,n}^{a,p} = w_{i,n-1}^{a,p}, \notag \\
             &\hspace{1.4cm}\forall n \in [M_k],~p \in \routeset, ~i \in \slotset{p}:i \neq d,~a \in \aircraftset, \label{eq:TFMP-constraint-aircraft-next-slot}\\
             &w_{v,n}^{a}=w_{v,n-1}^{a} + \sum_{p \in \routeset: d_p = v}(w_{d,n}^{a,p}-w_{d,n-1}^{a,p}), \notag \\
             &\hspace{1.4cm}\forall n \in [M_k],~v \in \vertset,~a \in \aircraftset, \label{eq:TFMP-constraint-aircraft-next-vertiport}\\
             &\sum_{p \in \routeset: o_p = v}w_{o,n}^{a,p} - w_{v,n-k_{\tau}}^{a} \leq 0, \notag \\
             &\hspace{1.4cm}\forall n \in \{k_{\tau}, \cdots, M_k\},~v \in \vertset,~a \in \aircraftset, \label{eq:TFMP-constraint-takeoff-lessthan-vertvisits}\\
             &\sum_{a \in \aircraftset}(w_{i,n}^{a,p}-w_{i,n-1}^{a,p}) + (w_{j,n}^{a,q}-w_{j,n-1}^{a,q}) \leq 1, \notag \\
             &\hspace{1.4cm}~\forall n \in [M_k],~ p,q \in \routeset,~ i \in \slotset{p}, j \in \slotset{q}: i=j, \notag\\
             &\hspace{4.88cm} i,j \neq o,d, \label{eq:TFMP-constraint-air-safety-over}\\
%
%
             &\sum_{a \in \aircraftset}\sum_{p \in \routeset: o_p=v}(w_{o,n}^{a,p}-w_{o,n-k_{\tau}}^{a,p}) \leq N_v, \notag\\
             &\hspace{1.4cm}~\forall n \in \{k_{\tau},\ldots,M_k\},~v \in \vertset, \label{eq:TFMP-constraint-takeoff-separation} \\
%
%
%
             &\sum_{a \in \aircraftset}\left(\sum_{p \in \routeset: o_p=v}(w_{o,n}^{a,p}-w_{o,n-1}^{a,p}) \right. \notag \\
             &\hspace{2cm} \left. +\sum_{q \in \routeset: d_q=v}(w_{d,n-1}^{a,q}-w_{d,n-k_{\tau}}^{a,q})\right) \leq N_v,\notag\\
             &\hspace{1.4cm}~\forall n \in \{k_{\tau},\ldots,M_k\},~v \in \vertset, \label{eq:TFMP-constraint-landing-separation} \\
%
%
             &w_{i,n}^{a,p},~ w_{v,n}^{a} \in \N_{0}, \notag \\
             &\hspace{1.4cm}~\forall n \in [M_k],~v \in \vertset,~p \in \routeset,~i \in \slotset{p},~a \in \aircraftset. \notag
        \end{align}
    \end{subequations}
\par 
Constraint \eqref{eq:TFMP-constraint-serviceall} ensures that all outstanding trip requests are serviced by the end of the cycle. Constraint \eqref{eq:TFMP-constraint-w-increasing} forces the decision variables $w_{i,n}^{a,p}$ to be non-decreasing in time. Constraint \eqref{eq:TFMP-constraint-aircraft-next-slot} guarantees that if vehicle $a$ occupies sector $i$ at some time $t \in (t_k+(n-1)\cruisetau,t_k+n\cruisetau]$, then it will occupy sector $i+1$ at time $t+\cruisetau$. Constraint \eqref{eq:TFMP-constraint-aircraft-next-vertiport} updates the number of visits to vertiport $v$ when vehicle $a$ lands on $v$, and implicitly forces the decision variables $w_{v,n}^{a}$ to be non-decreasing in time. Constraint \eqref{eq:TFMP-constraint-takeoff-lessthan-vertvisits} ensures that the number of takeoffs from vertiport $v$ is no more than the number of visits to vertiport $v$. Constraint \eqref{eq:TFMP-constraint-air-safety-over} ensures the airborne safety margins by allowing at most one UAM vehicle occupying any overlapping sector at any time. Constraint \eqref{eq:TFMP-constraint-takeoff-separation} ensures that the takeoff separation is satisfied at every vertiport by allowing at most $N_v$ takeoffs during any $\takeofftau$ time window. Similarly, \eqref{eq:TFMP-constraint-landing-separation} ensures that the landing separation is satisfied at every vertiport by restricting the number of landings and immediate takeoffs to at most $N_v$ during any $\takeofftau$ time window. 
\par
We also need additional constraints to take into account UAM vehicles' battery limitations. Let $E_p$ be the rate of battery consumption while flying route $p$, which is calculated as the sum of the battery consumption required for takeoff, cruise, and landing. Let $E_{n}^{a}$ be vehicle $a$'s state of charge at time $n$, and let $u_{v, n}^{a}$ denote the number of times vehicle $a$ has undergone re-charging at vertiport $v$ in the time interval $(t_k, t_k+n\cruisetau]$. In addition to the constraints in \eqref{eq:TFMP-constraints}, we require: 
\begin{subequations}\label{eq:TFMP-energy-constraints}
\begin{align}
    &E_{n}^{a} = E_{n-k_c}^{a} - \sum_{p \in \routeset}(w_{o,n}^{a,p}-w_{o,n-1}^{a,p})E_p \notag \\ 
    &\hspace{1.4cm}+ \sum_{v \in \vertset} (u_{v, n - k_c}^{a} - u_{v, n - k_c - 1}^{a}) (E_{\text{max}} - E_{n-k_c}^{a}), \notag \\
    &\hspace{1.4cm}~\forall n \in \{k_c + 1, \ldots, M_k\},~a \in \aircraftset, \label{eq:TFMP-constraint-energy-balance} \\
    &u_{v,n-1}^{a}-u_{v,n}^{a} \leq 0, \notag \\ 
    &\hspace{1.5cm}\forall n \in [M_k],~ v \in \vertset,~ a \in \aircraftset, \label{eq:TFMP-constraint-u-increasing}\\
    &E_{\text{min}} \leq E_{n}^{a} \leq E_{\text{max}},~ \notag \\
    &\hspace{1.4cm}~\forall n \in [M_k],~a \in \aircraftset, \label{eq:TFMP-constraint-energy-lower-upper} \\
%
%
    &u_{v, n}^{a} - u_{v, n - 1}^{a} \leq w_{v,n}^{a} - w_{v,n - 1}^{a}, \notag \\
    &\hspace{1.4cm}~\forall n \in [M_k],~v \in \vertset,~a \in \aircraftset, \label{eq:TFMP-constraint-energy-can-recharge} \\
    &\sum_{p \in \routeset} (w_{o, n}^{a, p} - w_{o, n - k_c}^{a, p}) \leq 1 - (u_{v, n - k_c}^{a} - u_{v, n - k_c - 1}^{a}), \notag \\
    &\hspace{1.4cm}~\forall n \in \{k_{c} + 1,\ldots,M_k\},~v \in \vertset,~a \in \aircraftset, \label{eq:TFMP-constraint-recharging-time} \\
    &u_{v,n}^{a} \in \N_{0}, \notag \\
    &\hspace{1.4cm}~\forall n \in [M_k],~v \in \vertset,~a \in \aircraftset, \notag
\end{align}
\end{subequations}
where $E_{\text{min}}$ and $E_{\text{max}}$ are the minimum and maximum allowed battery charge, respectively, and $k_c$ is an upper-bound on the number of time steps it takes to fully re-charge a UAM vehicle. Constraint \eqref{eq:TFMP-constraint-energy-balance} is the balance equation for vehicle $a$'s state of charge, and constraint \eqref{eq:TFMP-constraint-u-increasing} forces the decision variable $u_{v, n}^{a}$ to be non-decreasing in time. Constraint \eqref{eq:TFMP-constraint-energy-lower-upper} limits the minimum and maximum state of charge for vehicle $a$. Finally, constraint \eqref{eq:TFMP-constraint-energy-can-recharge} ensures that vehicle $a$ can be recharged at a vertiport only if it has visited that vertiport, while constraint \eqref{eq:TFMP-constraint-recharging-time} ensures that a takeoff can occur only after $k_c$ time steps. 
\par
The initial values $w_{i,0}^{a,p}$, $w_{v,0}^{a}$, and $u_{v, 0}^{a}$ at the start of a cycle are determined by the location and state of charge of vehicle $a$ at the end of the previous cycle. In particular, if vehicle $a$ has occupied sector $i$ of O-D pair $p$ at the end of cycle $k-1$, then $w_{v,0}^{a}=0$ for all $v \in \vertset$, $w_{j,0}^{a,p} = 1$ for sector $j=i$ and any other sector $j \in \slotset{p}$ that precedes sector $i$ along the UAM vehicle's route, and $w_{j,0}^{a,q} = 0$ for all $q \neq p$ and $j \in \slotset{q}$. 
\begin{remark}
In our formulation, we implicitly assume that each UAM vehicle has a passenger capacity of one. However, given the batch-based nature of the VertiSync policy, it can be easily extended to accommodate vehicles with higher capacities. Specifically, for a passenger capacity of $C$, the right-hand side of \eqref{eq:TFMP-constraint-serviceall} can be replaced with $\Queuelength{p}(t_k)/C$ to account for the fact that each vehicle can serve up to $C$ passengers per trip.
\end{remark}
\begin{remark}
A major difference between our formulation and the traditional TFMP for commercial airplanes is the inclusion of rebalancing. Specifically, the variables $w_{v, n}^{a}$, which track the number of visits to vertiport $v$ by vehicle $a$, along with their corresponding constraints are unique to our formulation. 
\end{remark}
\begin{remark}
Note that VertiSync only uses real-time information about the number of outstanding trip requests, and does not require any information about the arrival rate. This makes VertiSync a suitable option for an actual UAM network where the arrival rate is unknown or could vary over time.
\end{remark}
\end{definition}
\subsection{Size of VertiSync Formulation}\label{sec:problem-size}
In this section, we characterize the size of the optimization problem \eqref{eq:TFMP-objective}-\eqref{eq:TFMP-energy-constraints}, and we describe a pre-processing technique to reduce its size.
\par
Recall that $\vertcard{}$ denotes the number of vertiports, $\routecard$ denotes the number of O-D pairs, $\slotset{p}$ denotes the number of sectors associated with O-D pair $p$, and $\aircraftcard$ denotes the number of UAM vehicles. Moreover, $M_k$ is an integer that determines an upper-bound on the length of the $k$-th cycle. The total number of variables $w_{i,n}^{a,p}$ is $M_k \aircraftcard \sum_{p \in \routeset} |S_p|$, and the total number of variables $w_{v, n}^{a}$ and $u_{v, n}^{a}$ is $2M_k \aircraftcard \vertcard{}$. The number of constraints is upper-bounded by:
\begin{equation*}
    \routecard + M_k \left( 2 \aircraftcard + 1 \right) \sum_{p \in \routeset}|S_p| + M_k \left( 2 \vertcard{} + 5 \vertcard{} \aircraftcard + 2 \aircraftcard\right).
\end{equation*}
\par 
In order to get a sense of the size of the formulation, let us consider the following example:
\begin{example}\label{ex:number-of-constraints}
Consider the UAM network shown in Figure~\ref{fig:slots}. We have $\vertcard{} = 4$, $\routecard = 8$, $\sum_{p \in \routeset}|S_p| = 40$. Let $\aircraftcard = 5$ and $M_k = 100$. Then, the number of variables is 22,000 and the number of constraints is at most 55,808. 
\end{example}
\par
We can reduce the size of the optimization problem by concatenating some of the constraints. In particular, suppose that the route corresponding to O-D pair $p$ does not conflict with any other routes, except at the origin or destination vertiports. Then, we may remove the variables $w_{i,n}^{a,p}$, $i \neq o, d$, and their corresponding constraints and concatenate \eqref{eq:TFMP-constraint-aircraft-next-slot} into the following constraint:
\begin{equation*}
        w_{o,n-|S_p|}^{a,p}=w_{d,n}^{a,p}, ~\forall n \in \{|S_p|, \cdots, M_k\},~p \in \routeset, ~a \in \aircraftset. 
\end{equation*}
\par
Using the above pre-processing technique, the number of variables in Example~\ref{ex:number-of-constraints} is reduced to 14,000 and the number constraints to 34,708.

\subsection{VertiSync Throughput}
We next characterize the throughput of VertiSync. To do this, we introduce the notion of \emph{service vector}. A service vector is a $\routecard$-dimensional vector $\routingvec{}{}=(\routingvec{p}{})$ that specifies which O-D pairs can the UAM vehicles take off from and at what rate, so that the airborne safety margins and separation requirements are not violated. In particular, if $\routingvec{p}{} \neq 0$, then it means that UAM vehicles can safely takeoff at the rate $\routingvec{p}{}$ for O-D pair $p$. If $\routingvec{p}{} = 0$, then the takeoff rate for O-D pair $p$ is zero.
\par
Recall from the operational constraints in Section~\ref{section:op-constraints} that the takeoff rate from each vertipad is at most $1$ per $\takeofftau$ minutes, the takeoff rate from each vertiport is at most $1$ per $\cruisetau$, and $k_{\tau} = \takeofftau / \cruisetau$ is integer-valued. Therefore, if vertiport $v$ has $N_v$ vertipads, the takeoff rate from vertiport $v$ can be $0$, $\cruisetau/\takeofftau$, $2\cruisetau/\takeofftau$, \ldots, $\max\{N_v\cruisetau/\takeofftau, 1\}$ per time step. For example, a vector $\routingvec{}{i}$ with $\routingvec{p}{i} = \cruisetau/\takeofftau$ and $\routingvec{q}{i} = 0$ for all $q \neq p$ is a valid service vector since UAM vehicles can safely take off from O-D pair $p$ at the rate of $\cruisetau/\takeofftau$ vehicle per time step. 
We let $\safeschset$ be the set of all such non-zero service vectors. 
\begin{example}\label{ex:scheduling-vectors}
Consider the network in Figure~\ref{fig:slots}, which has $8$ O-D pairs $(1,3)$, $(1,4)$, $(2,3)$, $(2,4)$, $(3,1)$, $(4,2)$, $(1,2)$, and $(2,1)$ that we number from $1$ to $8$, respectively. Let the takeoff separation be $\takeofftau = 5$ minutes, and $\cruisetau = 0.5$ minutes and suppose that each vertiport has only one vertipad. Therefore, the takeoff rate from each vertiport is at most $\cruisetau/\takeofftau = 0.1$ vehicle per time step $\tau_c$. Moreover, note that O-D pairs $1$ and $4$ share a common link along their routes. However, if a UAM vehicle for O-D pair $1$ takes off at $t=0$, then another UAM vehicle for O-D pair $4$ can take off at $t=\cruisetau$ without violating the airborne safety margins, as it will not occupy the same sector with the first vehicle. Hence, $\routingvec{}{1}=(0.1,0,0,0.1,0,0,0,0)$ is a service vector. Similarly, $\routingvec{}{2}=(0,0.1,0.1,0,0,0,0,0)$ and $\routingvec{}{3}=(0,0.1,0,0,0,0,0,0)$ are two other service vectors. However, $(0.1,0.1,0,0,0,0,0,0)$ is not a service vector since UAM vehicles cannot simultaneously take off from vertiport $1$ at the rate of $0.1$ vehicle per time step.
\end{example}
\par
By using the service vectors described earlier, a feasible solution to the optimization problem \eqref{eq:TFMP-objective}-\eqref{eq:TFMP-energy-constraints} can be constructed as follows: (i) activate at most one service vector $\routingvec{}{i} \in \safeschset$ at any time, (ii) while $\routingvec{}{i}$ is active, schedule available UAM vehicles to take off at the rate $\routingvec{p}{i}$ for any O-D pair $p \in \routeset$, (iii) switch to another service vector in $\safeschset$ for servicing outstanding requests and/or rebalancing, provided that the safety margins with respect to airborne UAM vehicles from previous service vector are not violated, and (iv) repeat (i)-(iii) until all outstanding requests for the $k$-th cycle are serviced. 
\par
The next theorem provides an inner-estimate for the throughput of VertiSync when the number of UAM vehicles $\aircraftcard$ is sufficiently large and the following ``reversibility" assumption holds: 
\begin{assumption}\label{assumption:reversibility}
    (\textit{Reversibility}) For every service vector $\routingvec{}{i} \in \safeschset$, there exists a service vector $\routingvec{}{j} \in \safeschset$ such that for all $p \in \routeset$ with $\routingvec{p}{i} > 0$, $\routingvec{q}{j} = \routingvec{p}{i}$, where $q = (d_p,o_p)$ is the opposite O-D pair to the pair $p$. In other words, if all the O-D pairs in $\routingvec{p}{i}$ are ``reversed", then the resulting vector is also a service vector.
\end{assumption}
We define a service vector $\routingvec{}{i} \in \safeschset$ as ``symmetric" if, for all $p \in \routeset$, $\routingvec{q}{i} = \routingvec{p}{i}$, where $q = (d_p,o_p)$ is the opposite O-D pair to $p$. In other words, using the service vector $\routingvec{}{i}$, UAM vehicles can continuously take off at the same rate for both the O-D pair $p$ and its opposite pair $q$, without violating the safety margins and separation requirements. If a service vector is not symmetric but can be extended into a symmetric service vector, then we exclude it from the set $\safeschset$ without loss of generality to avoid redundancy. Specifically, given $\routingvec{}{i} \in \safeschset$, if there exists a symmetric service vector $\routingvec{}{j} \in \safeschset$ such that $\routingvec{p}{j} = \routingvec{p}{i}$ for all $p \in \routeset$ with $\routingvec{p}{i} > 0$, then we exclude $\routingvec{}{i}$ from $\safeschset$. Note that if a service vector is symmetric, it automatically satisfies the reversibility requirement in Assumption~\ref{assumption:reversibility} but the reverse is not true as seen in the following example.
\begin{figure}[t]
    \centering
    \includegraphics[width=0.3\textwidth]{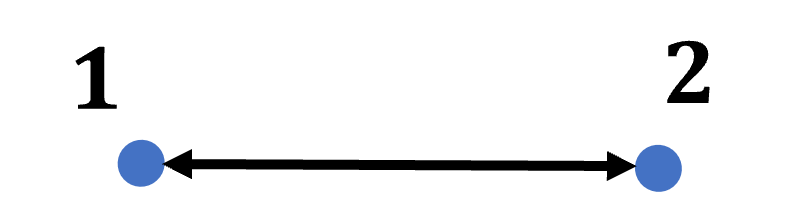}
    
    \vspace{0.1 cm}
    
    \caption{\sf A UAM network with $2$ vertiports (blue circles) and $2$ O-D pairs $(1,2)$ and $(2,1)$ sharing a single route (shown as a double-headed arrow).}
    \label{fig:graph-example-2}
\end{figure}
\begin{example}\label{ex:reversibe-non-symmetric-network}
Consider the simple network in Figure~\ref{fig:graph-example-2}, where there are only two O-D pairs $(1,2)$ and $(2,1)$ that share a single route. Suppose that each vertiport has only one vertipad, the takeoff separation is $\takeofftau = 5$ minutes, and $\cruisetau = 0.5$ minutes. Since there is only a single route, when a UAM vehicle is traveling in one direction, then no UAM vehicles can travel in the opposite direction. Therefore, there are only two service vectors $\routingvec{}{1}=(0.1,0)$ and $\routingvec{}{2}=(0,0.1)$. As can be seen, the network is reversible but not symmetric.
\end{example}
\begin{theorem}\label{thm:renewal-sufficient-reversable}
If the UAM network satisfies the reversibility Assumption~\ref{assumption:reversibility}, and the number of UAM vehicles satisfies
\begin{equation*}
    \aircraftcard \geq \max_{i \in [|\safeschset|]}\frac{\sum_{p \in \routeset}\routingvec{p}{i}}{\min_{\substack{p \in \routeset: \\ \routingvec{p}{i} > 0}}\routingvec{p}{i}},
\end{equation*}
then the VertiSync policy can keep the network under-saturated for demands belonging to the set 
\begin{equation*}
    D_{1}^{\circ}=\{\Arrivalrate{}: \Arrivalrate{} < \sum_{i = 1}^{\numsafesch}\frac{\routingvec{}{i}x_i}{1 + c_i},~ x_i \geq 0,~ i \in [\numsafesch],~ \sum_{i = 1}^{\numsafesch}x_i \leq 1\},
\end{equation*}
where 
\begin{equation}\label{eq:c_i-definition}
    \begin{aligned}
        c_i &= I_i +\\
        &(1 + I_i)\max_{p \in \routeset}\max\{\frac{T_p}{\cruisetau} + k_c - \frac{\aircraftcard}{\sum_{q}\routingvec{q}{i}}, \frac{T_p}{\cruisetau}I_{i}\}\frac{\sum_{p}\routingvec{p}{i}}{\aircraftcard}, \\
        I_i &= \begin{cases*}
           1 & ~$\routingvec{}{i}$~\text{non-symmetric} \\
           0 & ~$\routingvec{}{i}$~\text{symmetric}
        \end{cases*},
    \end{aligned}
\end{equation}
and the vector inequality $\Arrivalrate{} < \sum_{i = 1}^{\numsafesch}\routingvec{}{i} x_i/(1+c_i)$ is considered component-wise.
\end{theorem}
\begin{proof}
See Appendix~\ref{section:proof-renewal-sufficient}.
\end{proof}
A special case of Theorem~\ref{thm:renewal-sufficient-reversable} is when the UAM network is symmetric, i.e., all service vectors are symmetric, and the number of UAM vehicles is sufficiently large such that the inner maximum in $c_i$ is equal to zero.

\begin{corollary}\label{thm:renewal-sufficient-sym}
If the UAM network is symmetric, i.e., all service vectors in $\safeschset$ are symmetric, and the number of UAM vehicles satisfies 
\begin{equation*}
    \aircraftcard \geq \max_{i \in [|\safeschset]|}\sum_{p \in \routeset}\routingvec{p}{i}\max\{\frac{1}{\min_{\substack{p \in \routeset: \\ \routingvec{p}{i} > 0}}\routingvec{p}{i}}, 
    \max_{p \in \routeset}\frac{T_p}{\cruisetau} + k_c\},
\end{equation*}
then the VertiSync policy can keep the network under-saturated for demands belonging to the set 
\begin{equation*}
    D_{2}^{\circ}=\{\Arrivalrate{}: \Arrivalrate{} < \sum_{i = 1}^{\numsafesch}\routingvec{}{i} x_i,~\text{\small for}~ x_i \geq 0,~ i \in [\numsafesch],~ \sum_{i = 1}^{\numsafesch}x_i \leq 1\},
\end{equation*}
where the vector inequality $\Arrivalrate{} < \sum_{i = 1}^{\numsafesch}\routingvec{}{i} x_i$ is considered component-wise.
\end{corollary}
\begin{proof}
If the network is symmetric and the number of UAM vehicles satisfies
\begin{equation*}
    \aircraftcard \geq \max_{i \in [|\safeschset]|} \sum_{p \in \routeset}\routingvec{p}{i}\max_{p \in \routeset}\{\frac{T_p}{\cruisetau} + k_c\},
\end{equation*}
then the first term of the inner maximum in \eqref{eq:c_i-definition} becomes zero. Moreover, since $I_i = 0$, it follows that $c_i = 0$ for all $i \in [\numsafesch]$. The result then follows from Theorem~\ref{thm:renewal-sufficient-reversable}.
\mpcommentout{
The proof can be found in the arXiv version of this paper. 
}
\end{proof}
\begin{example}\label{ex:illustration-of-sufficient-theorem}
(\textbf{Example \ref{ex:reversibe-non-symmetric-network} cont'd}) Consider again the network in Figure~\ref{fig:graph-example-2}, where we number the two O-D pairs $(1,2)$ and $(2, 1)$ as $1$ and $2$, respectively. Let $T_1 = T_2 = 8$ minutes, and $k_c = 10$. Since both service vectors $\routingvec{}{1}=(0.1,0)$ and $\routingvec{}{2}=(0,0.1)$ are reversible and non-symmetric, we have $I_1 = I_2 = 1$. Moreover, we let $\aircraftcard = 32$, which satisfies the lower bound on $\aircraftcard$ from Theorem~\ref{thm:renewal-sufficient-reversable} since
\begin{equation*}
    \aircraftcard > \max_{i = 1, 2} \frac{0.1 + 0}{0.1} = 1.
\end{equation*}
\par
We now calculate $c_1$ and $c_2$ from Theorem~\ref{thm:renewal-sufficient-reversable} as follows:
\begin{equation*}
\begin{aligned}
    c_1 &= c_2 = 1 + 2\max_{p \in \routeset}\max\{\frac{8}{0.5} + 10 - \frac{32}{0.1}, \frac{8}{0.5}\}\frac{0.1}{32} \\
    &= 1.1.
\end{aligned}
\end{equation*}
Therefore, the set $D_{1}^{\circ}$ from Theorem~\ref{thm:renewal-sufficient-reversable} is:
\begin{equation*}
    \begin{aligned}
        &D_{1}^{\circ}=\{(\Arrivalrate{1}, \Arrivalrate{2}): \Arrivalrate{1} < \frac{x_1}{21}, \Arrivalrate{2} < \frac{x_2}{21}, \\
        &\hspace{1.7cm}x_1, x_2 \geq 0,~ x_1 + x_2 \leq 1\}.
    \end{aligned}
\end{equation*}
\end{example}
\subsection{Fundamental Limit on Throughput}
In this section, we provide an outer-estimate for the throughput of any conflict-free takeoff scheduling policy. A conflict-free policy is a policy that guarantees before takeoff that each UAM vehicle's entire route will be clear and a vertipad will be available for landing. 
\par
Any conflict-free policy uses the service vectors in $\safeschset$, either explicitly or implicitly, to schedule the UAM vehicles.
\mpcommentout{
This is done by activating one or multiple service vectors from $\safeschset$, scheduling the aircraft to take off at the rates specified by the activated service vectors, and switching between service vectors provided that the safety margins and separation requirements are not violated after switching.
}
Although it is possible for a conflict-free policy to activate multiple service vectors at any time, we may restrict ourselves to policies that activate at most one service vector from $\safeschset$ at any time. To justify this, we note that by activating at most one service vector at a time and rapidly switching between service vectors while ensuring that airborne safety margins are not violated, it is possible to achieve an exact or arbitrarily close approximation of any conflict-free schedule. 
\par
The next result provides a fundamental limit on the throughput of any conflict-free policy. This limit does not account for rebalancing or the number of available UAM vehicles.
\begin{theorem}\label{thm:necessary}
If a conflict-free policy $\pi$ keeps the network under-saturated, then the demand must belong to the set 
\begin{equation*}
\begin{aligned}
    D=\{\Arrivalrate{}: \Arrivalrate{} \leq \sum_{i = 1}^{\numsafesch}\routingvec{}{i} x_i,~\text{\small for}~ x_i \geq 0, ~i \in [\numsafesch],~ \sum_{i = 1}^{\numsafesch}x_i \leq 1\},
\end{aligned}
\end{equation*}
where the vector inequality $\Arrivalrate{} \leq \sum_{i = 1}^{\numsafesch}\routingvec{}{i} x_i$ is considered component-wise.
\end{theorem}
\begin{proof}
See Appendix \ref{section:proof-necessary}.
\end{proof}
\begin{remark}
    If the UAM network is symmetric and the number of UAM aircraft meets the lower bound in Corollary~\ref{thm:renewal-sufficient-sym}, then $D_2^{\circ}$ in Corollary~\ref{thm:renewal-sufficient-sym} is equal to $D$ \footnote{More precisely, $D_2^{\circ} = D$ \emph{almost} everywhere.}. Consequently, in this case, VertiSync serves as a policy that maximizes throughput. 
\end{remark}
\begin{example}\label{ex:illustration-of-necessary-theorem}
(\textbf{Examples \ref{ex:reversibe-non-symmetric-network} and \ref{ex:illustration-of-sufficient-theorem} cont'd}) For the parameters given in the previous examples, the set $D$ from Theorem~\ref{thm:necessary} is:
\begin{equation*}
    \begin{aligned}
        &D=\{(\Arrivalrate{1}, \Arrivalrate{2}): \Arrivalrate{1} < \frac{x_1}{10}, \Arrivalrate{2} < \frac{x_2}{10}, \\
        &\hspace{1.7cm}x_1, x_2 \geq 0,~ x_1 + x_2 \leq 1\}.
    \end{aligned}
\end{equation*}
\par
As can be seen, $D_{1}^{\circ} \subset D$.
\end{example}

\section{Simulation Results}\label{section:simulation}
\begin{figure}[t]
    \centering
    \includegraphics[width=0.35\textwidth]{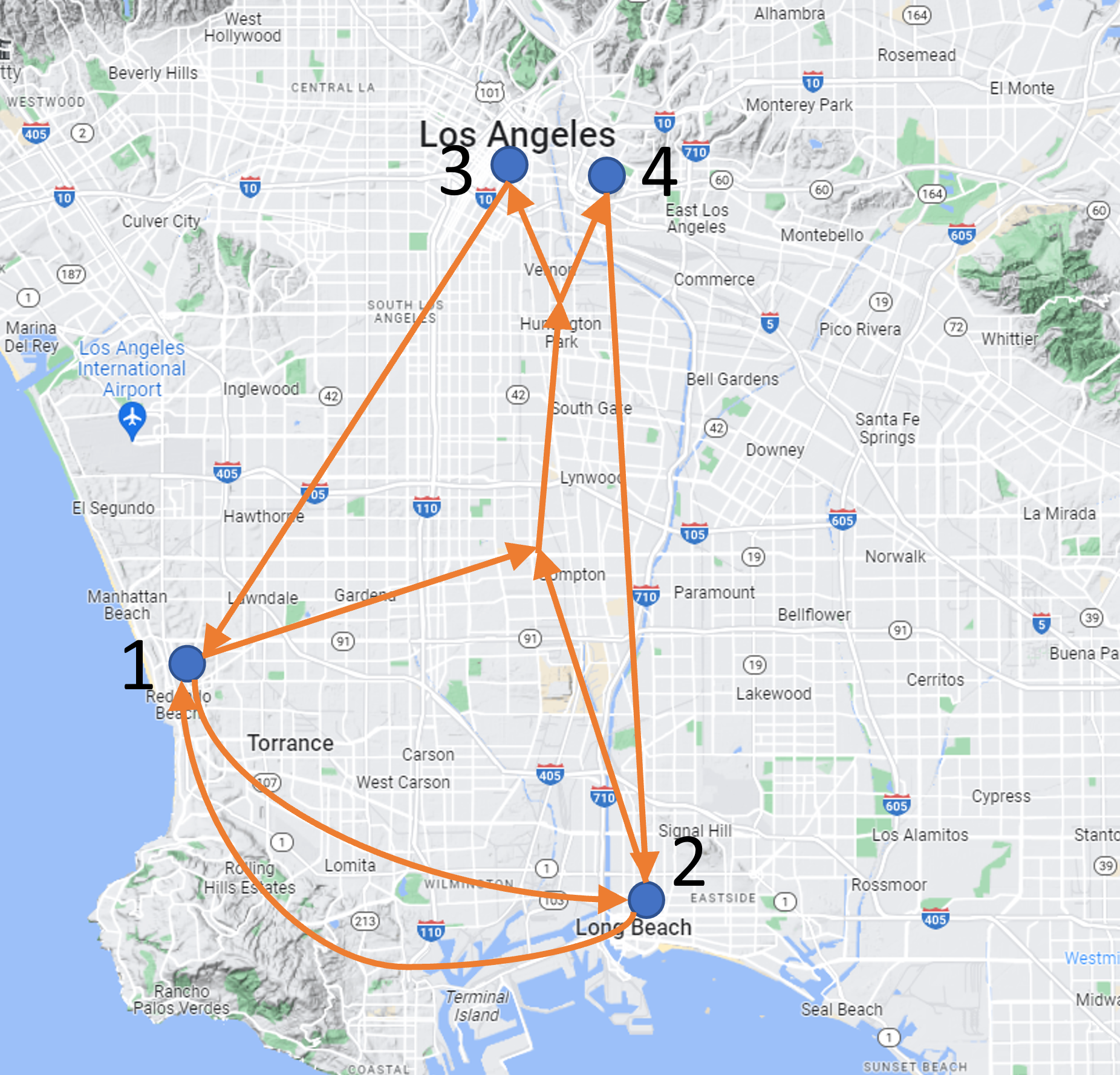}
    
    \vspace{0.1 cm}
    
    \caption{\sf A top-view sketch of a UAM network for Los Angeles. The blue circles show the vertiports and the orange arrows show the links.}
    \label{fig:LA-case-study}
\end{figure}
\begin{figure}[t]
    \centering
    \includegraphics[width=0.35\textwidth]{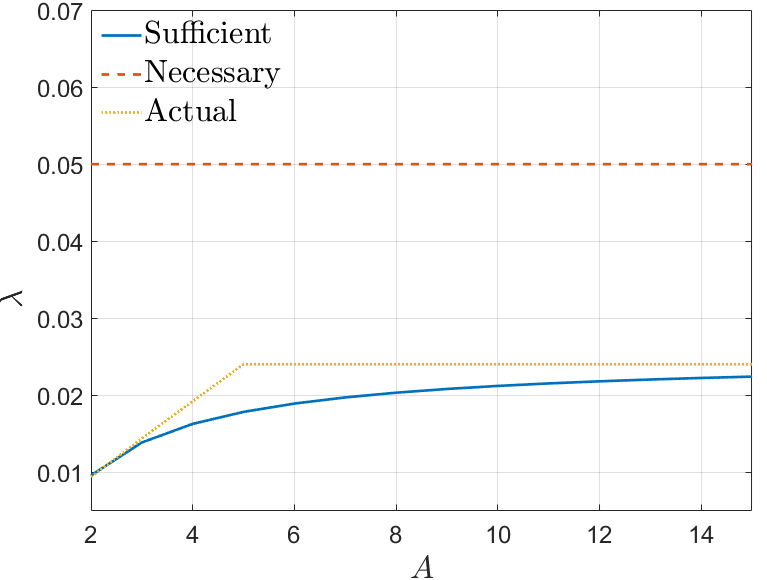}
    
    \vspace{0.1 cm}
    
    \caption{\sf The sufficient (from Theorem~\ref{thm:renewal-sufficient-reversable}), necessary (from Theorem~\ref{thm:necessary}), and actual (from simulations) bounds on $\Arrivalrate{}$.}
    \label{fig:bound-comparison}
\end{figure}
\begin{figure}[t]
    \centering
    \includegraphics[width=0.35\textwidth]{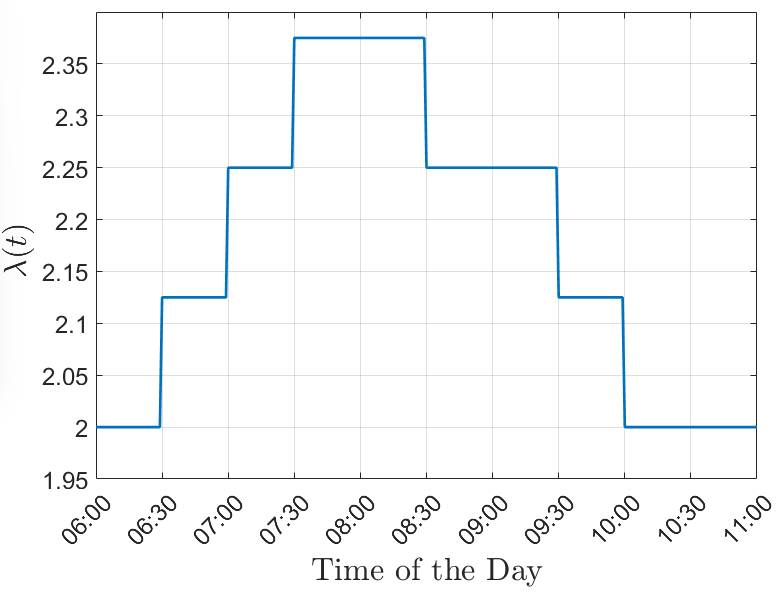}
    
    \vspace{0.1 cm}
    
    \caption{\sf The rate of trip requests per $\takeofftau$ minutes ($\Arrivalrate{}(t)$).}
    \label{fig:arrival-rate}
\end{figure}
In this section, we demonstrate the performance of the VertiSync policy and compare it with a heuristic scheduling policy from the literature. As a case study, we select the city of Los Angeles, which is anticipated to be an early adopter market due to severe road congestion, existing infrastructure, and mild weather \cite{vascik2017constraint}. All the simulations were performed using Python as the programming language, with optimization performed by Gurobi optimizer on a PC with Intel(R) Core(TM) i7-8700 processors, 3.2 GHz, 12 GB RAM with 64-bit Windows OS.
\par
\begin{figure*}[t]
    \centering
    \includegraphics[width=0.9\textwidth]{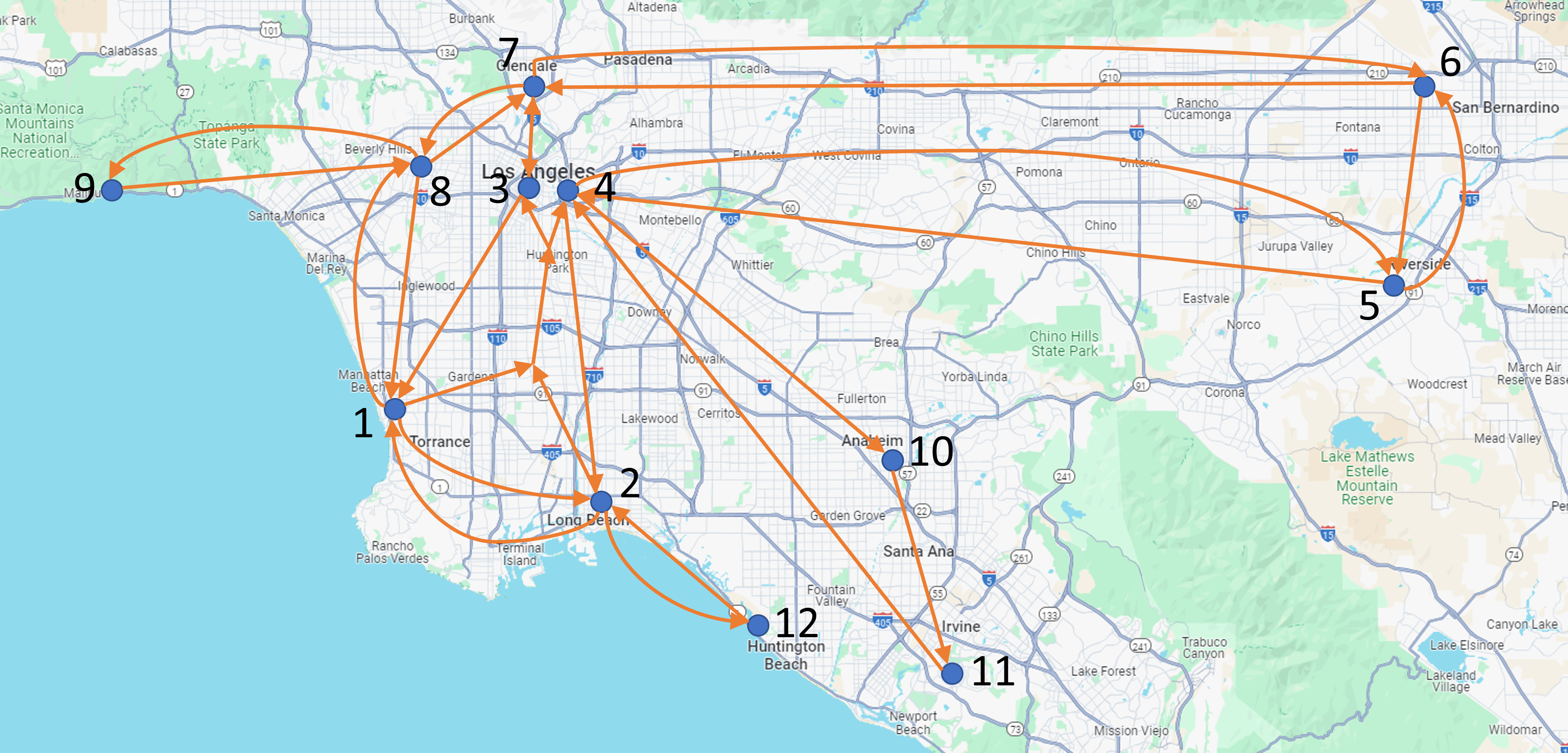}
    
    \vspace{0.1 cm}
    
    \caption{\sf A top-view sketch of an expanded UAM network for Los Angeles with $12$ vertiports and $27$ O-D pairs.}
    \label{fig:LA-case-study-large}
\end{figure*}
\begin{figure}[t]
    \centering
    \includegraphics[width=0.35\textwidth]{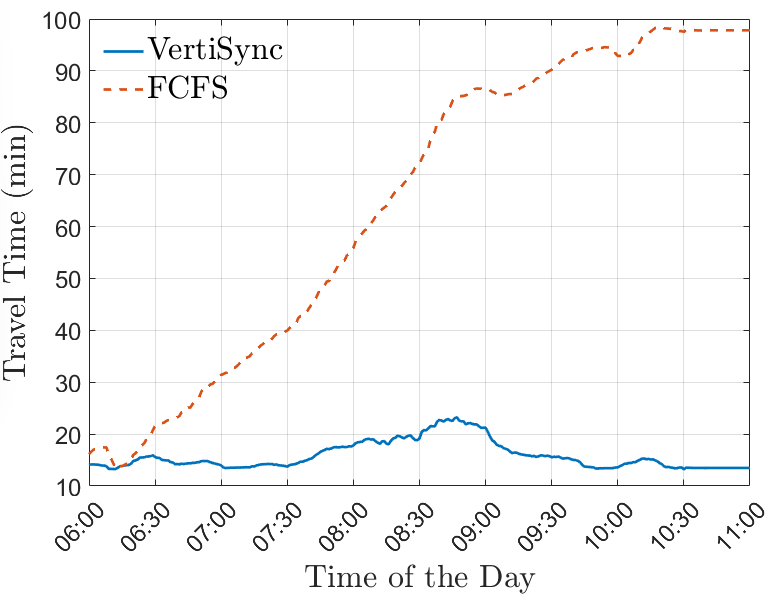}
    
    \vspace{0.1 cm}
    
    \caption{\sf The travel time under the VertiSync and FCFS policies for the demand $\Arrivalrate{}(t)$.}
    \label{fig:travel-time-comparison}
\end{figure}
\begin{figure}[t]
    \centering
    \includegraphics[width=0.35\textwidth]{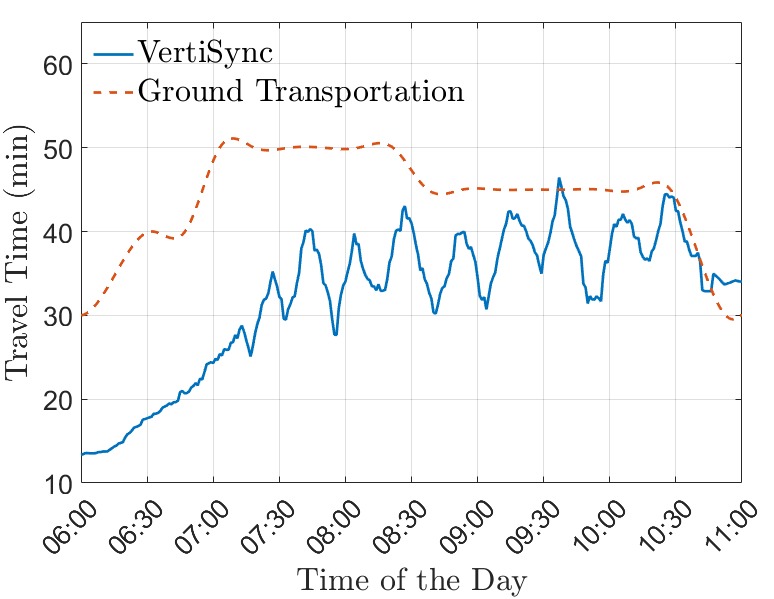}
    
    \vspace{0.1 cm}
    
    \caption{\sf The travel time under the VertiSync policy when the demand is increased to $1.2\Arrivalrate{}(t)$ (over-saturated regime), and the ground transportation travel time.}
    \label{fig:travel-time-ground}
\end{figure}
\subsection{Comparison of Theoretical Bounds}
In this section, we evaluate the ``sufficient" and ``necessary" under-saturation bounds given by Theorems~\ref{thm:renewal-sufficient-reversable} and \ref{thm:necessary}, respectively. We compare these bounds against the one obtained from simulations.

\par
We consider the Los Angeles network shown in Figure~\ref{fig:LA-case-study}, which consists of four vertiports located in Redondo Beach (vertiport $1$), Long Beach (vertiport $2$), and the Downtown Los Angeles area (vertiports $3$ and $4$). The choice of vertiport locations is adopted from \cite{vascik2017constraint}. Each vertiport is assumed to have $1$ vertipad. This network has eight O-D pairs $(1,3)$, $(1,4)$, $(2,3)$, $(2,4)$, $(3,1)$, $(4,2)$, $(1,2)$, and $(2,1)$, which we number from $1$ to $8$, respectively. We let the takeoff and landing separations $\takeofftau$ be $5$~[min], and the sector separation $\cruisetau$ be $0.5$~[min]. To simplify the simulations and without loss of generality, we assume that a UAM vehicle gets fully re-charged during the $\takeofftau$ period allocated for takeoff and landing operations. Removing this assumption would only result in a shift in the plots. The flight times for O-D pairs $(1,2)$ and $(2,1)$ are assumed to be $5$~[min], and for the rest of the O-D pairs are $8$~[min]. We let the trip requests for O-D pairs $1 - 4$ follow a Poisson process with the same rate $\Arrivalrate{}$. The demand for other O-D pairs is set to zero. 
\par
Given the number of vertiports and vertipads, this network has $40$ service vectors. However, the only service vectors that play a role in computing the sufficient and necessary bounds are $\routingvec{}{1} = (0.1, 0, 0, 0.1, 0, 0, 0, 0)$ and $\routingvec{}{2} = (0, 0.1, 0.1, 0, 0, 0, 0, 0)$. Using these service vectors, the sufficient and necessary bounds are calculated and shown in Figure~\ref{fig:bound-comparison} for different number of UAM vehicles. Note that the necessary bound found in Theoreom~\ref{thm:necessary} may be conservative. Therefore, we have also plotted the actual under-saturation bound in Figure~\ref{fig:bound-comparison}. The actual under-saturation bound is the fastest rate at which passengers can be serviced for O-D pair 1. Note that unlike the necessary bound from Theorem~\ref{thm:necessary}, the actual bound depends on the number of UAM vehicles. As can be seen from Figure~\ref{fig:bound-comparison}, the sufficient and actual bounds are less conservative than the necessary bound, as the necessary condition in Theorem~\ref{thm:necessary} does not take into account the number of UAM vehicles.

\subsection{Comparison with First-Come First-Serve Policy}\label{section:simulation-comparison}
We next evaluate travel time under our policy and the \emph{First-Come First-Serve} (FCFS) policy \cite{pradeep2018heuristic}. The FCFS policy is a heuristic policy which schedules the trip requests in the order of their arrival at the earliest time that does not violate the safety margins and separation requirements.
\par 
We again consider the Los Angeles network from the previous section, with each vertiport having $10$ vertipads. We let the number of UAM vehicles be $\aircraftcard = 32$, and assume that all of them are initially located at vertiport $1$. We let the takeoff, landing, and sector separations be the same as the previous section. Similar to the previous section and without loss of generality, we assume that a UAM vehicle gets re-charged during the $\takeofftau$ period allocated for takeoff and landing operations. We simulate this network during the morning period from 6:00-AM to 11:00-AM, during which the majority of demand originates from vertiports $1$ and $2$ and ends in vertiports $3$ and $4$. We let the trip requests for O-D pairs $1-4$ follow a Poisson process with a piece-wise constant rate $\Arrivalrate{}(t)$. The demand for other O-D pairs is set to zero during the morning period. With a slight abuse of notation, we scale $\Arrivalrate{}(t)$ to represent the number of trip requests per $\takeofftau$ minutes. From Theorem~\ref{thm:necessary}, given $\Arrivalrate{}(t) = \Arrivalrate{}$, the necessary condition for the network to remain under-saturated is that $\Arrivalrate{} \leq \takeofftau/4\cruisetau = 2.5$ trip requests per $\takeofftau$ minutes, i.e., $\rho := 4\Arrivalrate{}\cruisetau/\takeofftau \leq 1$. Figure~\ref{fig:arrival-rate} shows $\Arrivalrate{}(t)$, where we have considered a heavy demand between 7:00-AM to 9:30-AM to model the morning rush hour, i.e., $\rho(t) = 4\Arrivalrate{}(t)\cruisetau/\takeofftau \in [0.9,1)$ between 7:00-AM to 9:30-AM. 
\par
For the above demand and a random simulation seed, $518$ trips are requested during the morning period from which the FCFS policy services $411$ before 11:00~AM while the VertiSync policy is able to service all of them. Figure~\ref{fig:travel-time-comparison} shows the passenger travel time, which is computed by averaging the travel time of all trips requested within each $10$-minute time interval. The travel time of a trip is computed from the moment that trip is requested until it is completed, i.e., reached its destination. As can be seen, VertiSync is able to keep the network under-saturated, while the FCFS policy fails to do so, due to its greedy use of the vertipads and UAM airspace which is inefficient. 
\par
We next evaluate the demand threshold at which travel time under VertiSync becomes comparable to ground transportation. Figure~\ref{fig:travel-time-ground} shows the travel time under VertiSync when the demand is increased to $1.2\Arrivalrate{}(t)$. By Theorem~\ref{thm:necessary}, the network is in the over-saturated regime from 6:30-AM to 10:00-AM since $1.2\Arrivalrate{}(t) > 2.5$. However, as shown in Figure~\ref{fig:travel-time-ground}, the travel time is still less than the ground travel time during the morning period. The ground travel times are collected using the Google Maps service from 6:00-AM to 11:00-AM on Thursday, May 19, 2023 from Long Beach to Downtown Los Angeles (The travel times from Redondo Beach to Downtown Los Angeles were similar). 
\par

\mpcommentout{
\begin{figure}[t]
    \centering
    \includegraphics[width=0.35\textwidth]{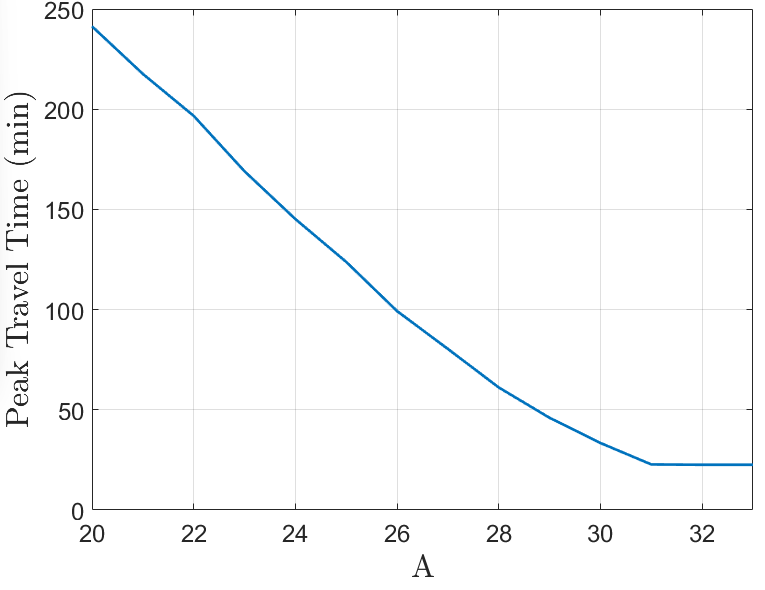}
    
    \vspace{0.1 cm}
    
    \caption{\sf The effect of the number of UAM vehicles $\aircraftcard$ on peak travel time under the VertiSync policy and demand $\Arrivalrate{}(t)$.}
    \label{fig:travel-time-fleetsize}
\end{figure}
Finally, we evaluate the effect of the number of UAM vehicles $\aircraftcard$ on the peak travel time under the VertiSync policy. Figure~\ref{fig:travel-time-fleetsize} shows the peak travel time versus the number of UAM vehicles, where the peak travel time is averaged over $1000$ simulation rounds with different seeds. For the demand $\Arrivalrate{}(t)$, the necessary number of UAM vehicles to keep the network under-saturated is $\aircraftcard = 31$. Moreover, the travel time does not improve as $\aircraftcard$ exceeds $32$. 
}
\subsection{Computation Results}\label{section:simulation-computation}
In this section, we present the computational experience
with the optimization problem \eqref{eq:TFMP-objective}-\eqref{eq:TFMP-energy-constraints}. We consider an expansion of the Los Angeles network from previous section to the network shown in Figure~\ref{fig:LA-case-study-large}. The network consists of $12$ vertiports and $27$ O-D pairs, with vertiport locations adopted from \cite{vascik2017constraint}. In this network, O-D pairs $(3, 7)$ and $(7, 3)$ are assumed to share a single route similar to Example~\ref{ex:reversibe-non-symmetric-network}. Moreover, note that even though $(11, 4)$ is an O-D pair, its opposite direction is not. Therefore, for rebalancing purposes, a UAM vehicle needs to make a stop at the intermediate vertiport $10$. 
\par
We let the number of vertipads, takeoff, landing, and airborne separations be the same as the previous section. We also let the number of UAM vehicles be $64$, and assume that they are evenly spread across the vertiports at the start of the cycle. We only consider a single cycle with a duration upper-bounded by $75$ minutes, i.e., $M_k \cruisetau = 75$ minutes in \eqref{eq:TFMP-objective}. Similar to the previous section, we assume that a UAM vehicle gets re-charged during the $\takeofftau$ period assigned to takeoff/landing operations.
\par
With the input data described above, the model has on the order of $2.6$ million constraints and $1.4$ million decision variables after applying the pre-processing technique in Section~\ref{sec:problem-size}. Theses numbers are about 5 times larger than the number of constraints and decision variables used in TFMP for national size instances in the United States \cite{bertsimas2011integer}. Given the size of the problem, we have taken advantage of the capability of Gurobi to stop after finding a good solution with optimality gap of $1$\%. However, as shall be seen, Gurobi is able to find optimal solutions within a reasonable time in all cases.
\par 
\begin{table}[th]
 \caption{Computational results for symmetric demand.}
 \vspace{0.2cm}
 \label{table:computational-results-sym}
\centering 
 \begin{tabular}{c c c c c} 
 \# Trip Requests & \begin{tabular}{@{}c@{}}CPU Time \\ (sec)\end{tabular}    &  \begin{tabular}{@{}c@{}}Cycle Len. \\ (min)\end{tabular} & O.F. & Gap (\%) \\ [1ex] 
 \hline  
 27 & 185.0 & 9 & 238 & 0.00 \\ 
 81 & 185.8 & 16 & 678 & 0.00 \\
 135 & 203.6 & 61 & 1,130 & 0.00 \\
 270 & 426.6 & 75 & 2,316 & 0.00 \\
 324 & 688.6 & 75 & 2,798 & 0.00 \\
 378 & Infeasible & - & - & - \\ [1ex] 
 \hline

 \end{tabular}
\end{table}

\begin{table}[th]
 \caption{Computational results for asymmetric demand.}
  \vspace{0.2cm}
 \label{table:computational-results-asym}
\centering 
 \begin{tabular}{c c c c c} 
 \# Trip Requests & \begin{tabular}{@{}c@{}}CPU Time \\ (sec)\end{tabular}    &  \begin{tabular}{@{}c@{}}Cycle Len. \\ (min)\end{tabular} & O.F. & Gap (\%) \\ [1ex] 
 \hline  
 27 & 179.8 & 75 & 270 & 0.00 \\ 
 83 & 261.0 & 43 & 1,146 & 0.00 \\
 125 & 347.7 & 75 & 2,078 & 0.34 \\
 173 & 1,415.1 & 75 & 2,794 & 0.00 \\ 
 215 & Infeasible & - & - & - \\ [1ex] 
 \hline

 \end{tabular}
\end{table}
We consider two cases for how the demand is spread across the network; symmetric and asymmetric. In the symmetric case, the number of trip requests are spread evenly across all O-D pairs. In the asymmetric case, $80 \%$ of the demand originates from vertiports $1, 2, 5, 6,$ and $11$ and ends in vertiports $3, 4,$ and $7$. The computational results for both cases are reported in Tables~\ref{table:computational-results-sym} and \ref{table:computational-results-asym}. The third column in each table shows the actual cycle length, and the fourth column shows the Objective Function (O.F.) value. It is clear from the results that Gurobi can compute the optimal solution within a reasonable time in all but two cases, with an average CPU time of 337.9 seconds for symmetric demand, and 300 seconds for asymmetric demand. The two infeasible cases arise because $M_k \cruisetau$ is too tight an upper bound for all trip requests to be processed within that time.
\par
We observe that the computational time for the asymmetric demand is generally longer than the symmetric case. This can be explained by noting that if several UAM vehicles need to fly the same O-D pair during the same time period, then there are several orderings in which they can do so without changing the value of the objective function. A second observation is that the feasible region for the asymmetric demand is much smaller than the symmetric demand. This is due to the increasing level of congestion in vertiports with high demand, which prevents flights to occur simultaneously. 
\par
A third observation is that the computational time tends to degrade when we are closer to the infeasiblity border. Indeed, by accepting a larger optimality gap, say $3\%$, the algorithm is able to compute a good quality solution much faster. For example, in the symmetric demand case with $324$ trip request, the computational time is reduced to $430.1$ seconds.

\section{Conclusion}\label{section:conclusion}
In this paper, we provided a conflict-free takeoff scheduling policy for on-demand UAM networks and analyzed its throughput. We conducted a case study for the city of Los Angeles and showed that our policy significantly improves travel time compared to a first-come first-serve policy. We also showed that our policy is computationally viable even for large instances of the problem. The next step in our research is to reduce computation
times further by using methods such as column generation \cite{balakrishnan2014optimal}, or by approximating the optimal solution using feed-forward neural networks \cite{bertsimas2022online}. We also plan to implement our policy in a high-fidelity air traffic simulator to study the effects of UAM vehicle dynamics and weather conditions on the performance of our policy.

\bibliographystyle{ieeetr}
\bibliography{References_main}

\appendices

\section{Proof of Theorem \ref{thm:renewal-sufficient-reversable}}\label{section:proof-renewal-sufficient}
Consider the $(k+1)$-th cycle. We first construct a feasible solution to the optimization problem \eqref{eq:TFMP-objective}-\eqref{eq:TFMP-energy-constraints} by using the service vectors in $\safeschset$. Consider the Linear Program (LP) 
\begin{equation}\label{eq:linear-program}
    \begin{aligned}
        \text{Minimize}&~\sum_{i=1}^{\numsafesch}K_i \\
        \text{Subject to}&~ \sum_{i=1}^{\numsafesch}\routingvec{}{i}K_i \geq \Queuelength{}(t_{k+1}), \\
        &~ K_i \geq 0,~ i \in [\numsafesch],
    \end{aligned}
\end{equation}
where the inequality $\sum_{i=1}^{\numsafesch}\routingvec{}{i}K_i \geq \Queuelength{}(t_{k+1})$ is considered component-wise. Let $K_{i}^{*}$, $i \in [\numsafesch]$, be a feasible solution to \eqref{eq:linear-program}. A feasible solution to the optimization problem \eqref{eq:TFMP-objective}-\eqref{eq:TFMP-energy-constraints} can be constructed as follows: 
\begin{enumerate}
    \item Choose a service vector $\routingvec{}{i} \in \safeschset$ with $K_{i}^{*} > 0$, and, before activating $\routingvec{}{i}$, distribute the UAM vehicles in the system so that for any $p \in \routeset$ with $\routingvec{p}{i} > 0$, there are
    \begin{equation*}
       \aircraftcard_p := \frac{\routingvec{p}{i}}{\sum_{p' \in \routeset}\routingvec{p'}{i}}\aircraftcard 
    \end{equation*}
    vehicles at vertiport $o_p$. Note that $\aircraftcard_p \geq 1$ by the assumption on the minimum number of UAM vehicles. The initial distribution of UAM vehicles takes at most $\aircraftcard \overline{T} / \cruisetau$ time steps, where $\overline{T} = \max_{p \in \routeset}T_p$. 
    
    \item Once the initial distribution is completed and the airspace is empty, we would like to activate $\routingvec{}{i}$ for $K_{i}^{*}+k_{\tau}$ time steps to service outstanding trip requests for route $p$. However, to ensure that a UAM vehicle is available at vertiport $o_p$ when $\routingvec{}{i}$ is active, additional time for rebalancing may be required. Let $C_i$ denote an upper bound on this additional rebalancing time. If $\routingvec{}{i}$ is symmetric, then 
    \begin{equation}\label{eq:rebalancing-bound-sym}
        \begin{aligned}
            &C_i =\max_{\substack{p \in \routeset:\\ \routingvec{p}{i} > 0}} \max \{\frac{T_p}{\cruisetau} + k_c - \frac{\aircraftcard_p}{\routingvec{p}{i}}, 0\}\frac{\routingvec{p}{i}(K_{i}^{*}+k_{\tau})}{\aircraftcard_p} \\
            &\hspace{0.3cm}= C_i(K_{i}^{*}) + C_i^{'},
        \end{aligned}
    \end{equation}
    where $C_i(K_{i}^{*})$ is the part that depends on $K_{i}^{*}$ and $C_i^{'}$ is the part that does not depend on $K_{i}^{*}$. The right hand side of \eqref{eq:rebalancing-bound-sym} is the time it takes for a UAM vehicle from the opposite direction $q = (d_p, o_p)$ to reach vertiport $o_p$ ($T_p/\cruisetau$) and gets recharged ($k_c$), subtracted by the amount of time it takes for all available vehicles to depart from $o_p$ ($\aircraftcard_p/\routingvec{p}{i}$), multiplied by $\routingvec{p}{i}(K_{i}^{*}+k_{\tau})/\aircraftcard_p$, which is the total number of iterations needed for $\routingvec{}{i}$ to be used for $K_{i}^{*}+k_{\tau}$ time steps. On the other hand, if $\routingvec{}{i}$ is non-symmetric, then 
    \begin{equation}\label{eq:rebalancing-bound-non-sym}
        \begin{aligned}
            &C_i = K_{i}^{*} +\\
            & 2\max_{\substack{p \in \routeset:\\ \routingvec{p}{i} > 0}} \max \{\frac{T_p}{\cruisetau} + k_c - \frac{\aircraftcard_p}{\routingvec{p}{i}}, \frac{T_p}{\cruisetau}\}\frac{\routingvec{p}{i}(K_{i}^{*}+k_{\tau})}{\aircraftcard_p} \\
            &\hspace{0.3cm}= C_i(K_{i}^{*}) + C_i^{'}.
        \end{aligned}
    \end{equation}
    \par
    The right hand side of \eqref{eq:rebalancing-bound-non-sym} is calculated similar to \eqref{eq:rebalancing-bound-sym}. Note that if $\routingvec{}{i}$ is symmetric, UAM vehicles can be rebalanced while $\routingvec{}{i}$ is active. However, if $\routingvec{}{i}$ is not symmetric, it must be deactivated first before all UAM vehicles can be rebalanced to vertiport $o_p$.  

    \item Once step 2 is completed and the airspace is empty, repeat steps 1 and 2 for another vector in $\safeschset$. The amount of time it takes for the airspace to become empty at the end of step 2 is at most $\overline{T} / \cruisetau$ time steps. Once each service vector $\routingvec{}{i} \in \safeschset$ with $K_{i}^{*} > 0$ has been activated, $\sum_{i=1}^{\numsafesch}\routingvec{p}{i}K_{i}^{*}$ requests will be serviced for each O-D pair $p \in \routeset$. From the constraint of the LP \eqref{eq:linear-program}, $\sum_{i=1}^{\numsafesch}\routingvec{}{i}K_{i}^{*} \geq \Queuelength{}(t_{k+1})$, i.e., all the requests for the $(k+1)$-th cycle will be serviced and the cycle ends. 
\end{enumerate}
By combining the time each of the above steps takes, it follows that
\begin{equation}\label{eq:cycle-upper-bound}
    \Cyclelength(k+1) \leq \sum_{i=1}^{\numsafesch}(1+c_i)K_{i}^{*} + \sum_{i=1}^{\numsafesch}C_{i}^{'} + \frac{\numsafesch}{\cruisetau}(\aircraftcard \overline{T} + \overline{T}),
\end{equation}
where $c_i = C_i(K_{i}^{*}) / K_{i}^{*}$, which does not depend on $K_{i}^{*}$, and $\Cyclelength(k+1) = (t_{k+2} - t_{k+1})/\cruisetau$.
\par
Without loss of generality, we assume that the ordering by which $\routingvec{}{i}$'s are chosen at each cycle are fixed. Moreover, when a cycle ends, we postpone the start of the next cycle to when the airspace becomes empty, and we assume that this new start time is a multiple of $\cruisetau$. 
With these assumptions, we can cast the network as a discrete-time Markov chain with the state $\{\Queuelength{}(t_k)\}_{k \geq 1}$. Since the state $\Queuelength{}(t_k) = 0$ is reachable from all other states, and $\mathbb{P}\left(\Queuelength{}(t_{k+1}) = 0~|~ \Queuelength{}(t_k) = 0\right) > 0$, the chain is irreducible and aperiodic. Consider the function $f: \Z_{+}^{\routecard} \ra [0,\infty)$
\begin{equation*}
    \Lyap{\Queuelength{}(t_k)} = \Cyclelength^{2}(k),
\end{equation*}
where $\Z_{+}^{\routecard}$ is the set of $\routecard$-tuples of non-negative integers. Note that $\Cyclelength(k)$ is a non-negative integer from our earlier assumption that the cycle start times are a multiple of $\cruisetau$. We let $\Lyap{\Queuelength{}(t_k)} \equiv \Lyap{t_k}$ for brevity.
\par

We next show that
\begin{equation}\label{eq:cycle-decreases-expectation}
    \limsup_{n \ra \infty}\E{\left(\frac{\Cyclelength(k+1)}{\Cyclelength(k)}\right)^2~\middle|~ \Cyclelength(k)=n} < 1.
\end{equation}
\par
To show \eqref{eq:cycle-decreases-expectation}, let $\Cyclelength(k)=n$, and let $\overline{A_p}(t_k,t_{k+1})$ be the cumulative number of trip requests for the O-D pair $p \in \routeset$ during the time interval $[t_k,t_{k+1})$. Note that $\Queuelength{p}(t_{k+1})=\overline{A_p}(t_k,t_{k+1})$, which implies from the strong law of large numbers that, with probability one, 
\begin{equation*}
    \lim_{n \ra \infty}\frac{\Queuelength{p}(t_{k+1})}{n} = \Arrivalrate{p}.
\end{equation*}
\par
By the assumption of the theorem, $\Arrivalrate{} \in D_{1}^{\circ}$. Hence, with probability one, there exists $N' > 0$ such that for all $n > N'$ we have $\Queuelength{}(t_{k+1})/n \in D_{1}^{\circ}$. Since $D_{1}^{\circ}$ is an open set, for a given $n > N'$, there exists non-negative $x_1, x_2, \ldots, x_{\numsafesch}$ with $\sum_{i=1}^{\numsafesch}x_i < 1$ such that $\Queuelength{}(t_{k+1})/n < \sum_{i=1}^{\numsafesch}\routingvec{}{i}x_i / (1 + c_i)$, or equivalently, $\Queuelength{}(t_{k+1}) < \sum_{i=1}^{\numsafesch}\routingvec{}{i} n x_i / (1 + c_i)$. Thus,  if we let $K_i := nx_i/(1 + c_i)$, $i \in [\numsafesch]$, then, $\Queuelength{}(t_{k+1}) < \sum_{i=1}^{\numsafesch}\routingvec{}{i} K_i$, which implies that $K_i$'s are a feasible solution to the LP \eqref{eq:linear-program}. Moreover, $\sum_{i=1}^{\numsafesch}(1 + c_i)K_i < n$. Therefore, from \eqref{eq:cycle-upper-bound} and with probability one, it follows for all $n > N'$ that
\begin{equation*}
    \begin{aligned}
        \Cyclelength(k+1) &\leq \sum_{i=1}^{\numsafesch}(1 + c_i)K_{i}^{*} + \sum_{i=1}^{\numsafesch}C_{i}^{'} + \frac{\numsafesch}{\cruisetau}(\aircraftcard \overline{T} + \overline{T}) \\
        &< n + \sum_{i=1}^{\numsafesch}C_{i}^{'} + \frac{\numsafesch}{\cruisetau}(\aircraftcard \overline{T} + \overline{T}),
    \end{aligned}
\end{equation*}
which in turn implies, with probability one, that
\begin{equation}\label{eq:cycle-decreases-w.p.1}
    \limsup_{n \ra \infty}\left(\frac{\Cyclelength(k+1)}{n}\right)^2 < 1.
\end{equation}
\par
Finally, since the number of trip requests for each O-D pair is at most $1$ per $\cruisetau$ minutes, the sequence $\{\Cyclelength(k+1)/n\}_{n=1}^{\infty}$ is upper bounded by an integrable function. Hence, from \eqref{eq:cycle-decreases-w.p.1} and the Fatou's Lemma \eqref{eq:cycle-decreases-expectation} follows.
\par
We will now use \eqref{eq:cycle-decreases-expectation} to show that the network is under-saturated. Note that \eqref{eq:cycle-decreases-expectation} implies that there exists $\delta \in (0,1)$ and $N$ such that for all $n > N$ we have
\begin{equation*}
    \E{\left(\frac{\Cyclelength(k+1)}{\Cyclelength(k)}\right)^2 \middle|~ \Cyclelength(k) = n} < 1 - \delta,
\end{equation*}
which in turn implies that
\begin{equation*}
    \E{\Cyclelength^2(k+1) - \Cyclelength^2(k) \middle|~ \Cyclelength(k) > N} < -\delta \Cyclelength^2(k).
\end{equation*}
\par 
Furthermore, $\Queuelength{p}(t_k) \leq \Cyclelength(k) \leq \Cyclelength^2(k)$ for all $p \in \routeset$, where the first inequality follows from the fact that $t_{k+1}-t_{k} \geq \Queuelength{p}(t_k)\cruisetau$ for any O-D pair $p \in \routeset$. Therefore, $\E{\Cyclelength^2(k+1) - \Cyclelength^2(k) \middle|~ \Cyclelength(k) > N} < -\delta\|\Queuelength{}(t_k)\|_{\infty}$, where $\|\Queuelength{}\|_{\infty} = \max_{p}\Queuelength{p}$. Finally, if $\Cyclelength(k) \leq N$, then $\Cyclelength(k+1) \leq 2 N \routecard (\overline{T}/\cruisetau + k_{\takeofftau}) =: b$. Therefore, 
\begin{equation*}
\begin{aligned}
    \E{\Cyclelength^2(k+1) \middle|~ \Cyclelength(k) \leq N} &\leq b^2 \\
    & +\Cyclelength^2(k)-\delta\|\Queuelength{}(t_k)\|_{\infty},
\end{aligned}
\end{equation*}
where we have used $\delta\|\Queuelength{}(t_k)\|_{\infty} \leq \|\Queuelength{}(t_k)\|_{\infty} \leq \Cyclelength^2(k)$. Combining all the previous steps gives
\begin{equation*}
    \E{f(t_{k+1}) - f(t_k) \middle|~ \Queuelength{}(t_k)} \leq -\delta \|\Queuelength{}(t_k)\|_{\infty} + b^2\mathds{1}_{B},
\end{equation*}
where $B = \{\Queuelength{}(t_k): f(t_k) \leq N^2\}$ (a finite set). From this and the well-known Foster-Lyapunov drift criterion \cite[Theorem 14.0.1]{meyn2012markov}, it follows that $\limsup_{t \to \infty} \E{\Queuelength{p}(t)} < \infty$ for all $p \in \routeset$, i.e., the network is under-saturated.

\section{Proof of Theorem \ref{thm:necessary}}\label{section:proof-necessary}
We prove by contradiction. Suppose that some conflict-free policy $\pi$ keeps the network under-saturated but $\Arrivalrate{} \notin D$. Then, for any non-negative $x_1, x_2, \ldots, x_{\numsafesch}$ with $\sum_{i =1}^{\numsafesch}x_i \leq 1$, there exists some O-D pair $p \in \routeset$ such that $\Arrivalrate{p} > \sum_{i = 1}^{\numsafesch}\routingvec{p}{i} x_i$.
\par
Without loss of generality, we may assume that whenever the service vector $\routingvec{}{i}$ becomes active, it remains active for a time interval that is a multiple of $\cruisetau$. Given $k \in \N_{0}$, let $s_k := k\cruisetau$, and let $x_i(s_k)$ be the proportion of time that the service vector $\routingvec{}{i}$ has been active under policy $\pi$ up to time $s_k$. Then, $x_i := \limsup_{k \ra \infty}x_i(s_k) \geq 0$ for all $i \in [\numsafesch]$ and $\sum_{i = 1}^{\numsafesch}x_i \leq 1$. Therefore, there exists $p \in \routeset$ such that $\Arrivalrate{p} > \sum_{i = 1}^{\numsafesch}\routingvec{p}{i} x_i$. Note that when the service vector $\routingvec{}{i}$ is active, the trip requests for O-D pair $p$ are serviced at the rate of at most $\routingvec{p}{i}$. Hence, the number of trip requests for O-D pair $p$ that have been serviced by $\routingvec{}{i}$ up to time $s_k$ is at most $\routingvec{p}{i}x_i(s_k)s_k$. Let $\overline{A_p}(s_k) := \overline{A_p}(0,s_k)$ be the cumulative number of flight requests for O-D pair $p$ up to time $s_k$. We have
\begin{equation*}
    \Queuelength{p}(s_k) \geq \Queuelength{p}(0) + \overline{A_p}(s_k) - \sum_{i = 1}^{\numsafesch}\routingvec{p}{i} x_i(s_k) s_k,
\end{equation*}
which implies
\begin{equation*}
    \frac{\Queuelength{p}(s_k)}{s_k} \geq \frac{\Queuelength{p}(0)}{s_k} + \frac{\overline{A_p}(s_k)}{s_k} - \sum_{i = 1}^{\numsafesch}\routingvec{p}{i} x_i(s_k).
\end{equation*}
\par
By letting $k \ra \infty$, it follows from the strong law of large numbers that, with probability one,
\begin{equation*}
    \liminf_{k \ra \infty} \frac{\Queuelength{p}(s_k)}{k} \geq \Arrivalrate{p} - \sum_{i = 1}^{\numsafesch}\routingvec{p}{i} x_i.
\end{equation*}
\par
Since $\Arrivalrate{p} > \sum_{i = 1}^{\numsafesch}\routingvec{p}{i}x_i$, then, with probability one, $\liminf_{k \ra \infty} \Queuelength{p}(s_k)/k$ is bounded away from zero. Hence, 
\begin{equation*}
  \liminf_{k \ra \infty} \Queuelength{p}(s_k) = \infty.  
\end{equation*}
\par
Therefore, the expected number of flight requests for the O-D pair $p$ grows unbounded. This contradicts the network being under-saturated.

\end{document}